\documentclass[nonacm,sigconf]{acmart}
\pdfoutput=1
\usepackage{booktabs} 
\usepackage{courier}
\usepackage{graphicx,graphics}
\usepackage{url}
\usepackage{multirow}
\usepackage{color}
\usepackage{subfig}
\usepackage{balance}
\usepackage{algorithm}
\usepackage{amsmath}
\usepackage{algorithmic}
\usepackage{pgfplots}
\usepackage{bm}
\usepackage{physics}

\newenvironment{customlegend}[1][]{%
	\begingroup
	\csname pgfplots@init@cleared@structures\endcsname
	\pgfplotsset{#1}%
}{%
	\csname pgfplots@createlegend\endcsname
	\endgroup
}%
\def\addlegendimage{\csname pgfplots@addlegendimage\endcsname}

\DeclareMathOperator*{\argmax}{argmax}
\DeclareMathOperator*{\argmin}{argmin}

\newtheorem{theorem}{Theorem}[subsection]
\newtheorem{lemma}{Lemma}
\newtheorem{definition}{Definition}
\newtheorem{proposition}[theorem]{Proposition}
\newtheorem{claim}{Claim}

\usepackage{enumitem}

\newcommand{\gourab}[1]{\textcolor{blue}{GP:#1}}
\newcommand{\noteng}[1]{\textcolor{red}{NG:#1}}

\copyrightyear{2022}
\acmYear{2022}
\setcopyright{acmlicensed}\acmConference[WWW '22]{Proceedings of the ACM
	Web Conference 2022}{April 25--29, 2022}{Virtual Event, Lyon, France}
\acmBooktitle{Proceedings of the ACM Web Conference 2022 (WWW '22), April
	25--29, 2022, Virtual Event, Lyon, France}
\acmPrice{15.00}
\acmDOI{10.1145/3485447.3512136}
\acmISBN{978-1-4503-9096-5/22/04}

\begin{document}

\title{Scheduling Virtual Conferences Fairly: \\ Achieving Equitable Participant and Speaker Satisfaction}

\author{Gourab K. Patro}
\affiliation{\institution{IIT Kharagpur, India}\country{}}
\affiliation{\institution{L3S Research Center, Germany}\country{}}

\author{Prithwish Jana}
\affiliation{\institution{IIT Kharagpur, India}\country{}}

\author{Abhijnan Chakraborty}
\affiliation{\institution{IIT Delhi, India}\country{}}

\author{Krishna P. Gummadi}
\affiliation{\institution{MPI-SWS, Germany}\country{}}

\author{Niloy Ganguly}
\affiliation{\institution{IIT Kharagpur, India}\country{}}
\affiliation{\institution{L3S Research Center, Germany}\country{}}


\renewcommand{\shortauthors}{Patro et al. WWW'22}
\begin{abstract}
Recently, almost all conferences have moved to virtual mode due to the pandemic-induced restrictions on travel and social gathering. Contrary to in-person conferences, virtual conferences face the challenge of efficiently scheduling talks, accounting for the availability of participants from different timezones and their interests in attending different talks. A natural objective for conference organizers is to maximize efficiency, e.g., total expected audience participation across all talks. However, we show that optimizing for efficiency alone can result in an unfair virtual conference schedule, where  individual utilities for participants and speakers can be highly unequal. To address this, we formally define fairness notions for participants and speakers, and derive suitable objectives to account for them. As the efficiency and fairness objectives can be in conflict with each other, 
we propose a joint optimization framework that allows conference organizers to design schedules that balance (i.e., allow trade-offs) among efficiency, participant fairness and speaker fairness objectives. While the optimization problem can be solved using integer programming to schedule smaller conferences, we provide two scalable techniques to 
cater to bigger conferences. Extensive evaluations over multiple real-world datasets show the efficacy and flexibility of our proposed approaches.
\end{abstract}
\ccsdesc[300]{Theory of computation~Scheduling algorithms}
\keywords{Fair Conference Scheduling, Virtual Conference Scheduling}	
\maketitle
\section{Introduction}
Restrictions on travel and social gatherings to tackle the COVID-19 pandemic have forced almost all conferences to move online, and some of them may remain online in future due to the benefits online conferences offer. 
They are hugely economical due to reduced organizational costs,
and they foster inclusivity by significantly improving the scale and outreach~\cite{oxford2020conf}.
%
However, online conferences have their own set of challenges, such as, 
scaling up participation 
depends on stable and high-speed Internet in different regions; 
sometimes participants and speakers need to be trained on different conferencing tools 
for efficient participation~\cite{saliba2020getting}.
A big challenge in organizing an online conference is {\it optimal scheduling of the conference talks}. 
Online conferences usually have participants from all around the globe, unlike the physical conferences where participants assemble at a single place. Thus, traditional timezone-specific conference schedules (based on the timezone of the venue) are no longer suitable for online conferences, as the participants from the other parts of the globe will find it hard to attend (\cite{donaldson2020report,misa2020lessons}).
This demands for conference schedules to be {\it timezone-aware} instead of {\it timezone-specific}, and may stretch beyond the usual 7-8 hours a day, 
to cater to participants from different timezones.
In this paper, we focus on this conference scheduling problem and 
associated concerns about efficiency and fairness.
In conference scheduling, a natural objective for organizers would be to maximize an efficiency measure, such as the {\it total expected audience participation across all talks}---similar to the participation metrics used in prior literature on optimal meeting scheduling \cite{garrido1996multi,capek2008event,maheswaran2004taking,pino1998scheduling,chun2003optimizing}. 
%
While 
earlier works have focused on optimizing efficiency, optimizing for such objective alone in online conference settings can result in a schedule where
the level of satisfaction of individual participants 
may vary widely, and the expected exposure 
(audience size) at different talks can be highly 
skewed---leading to disparity in speaker satisfactions.
Note that the prior works on meeting scheduling \cite{capek2008event,maheswaran2004taking,chun2003optimizing} model only the participant satisfaction, but have no concept of speaker satisfaction.
Intuitively, a participant would be less satisfied if her favorite talks are scheduled in timeslots unfavorable for her, and similarly a speaker would be less satisfied if her talk is scheduled in a timeslot that adversely 
limits the expected audience or crowd at her talk.
Thus the organizers need to consider fairness along with efficiency.

We formally define the conference scheduling problem in sec-\ref{sec:prelim} alongside suitable measures of participant satisfaction, speaker satisfaction, and efficiency. 
Intuitively, a schedule would be fair if it ensures equity of satisfaction among individual participants as well as among the speakers.
We formally define the fairness notions in sec-\ref{subsec:p_fairness} and sec-\ref{subsec:s_fairness}.
%
%
%
%
We further show that it may be impossible to maximize efficiency, participant and speaker fairness objectives simultaneously as there are fundamental tensions among these objectives (more details in sec-\ref{subsec:tensions}) --- 
optimizing one objective could cause losses in other objectives.
Thus, we propose a joint optimization framework (sec-\ref{subsec:joint_opt}) which allows conference organizers to design schedules that balance (i.e., allow trade-offs) among the efficiency, participant fairness and speaker fairness objectives.
We show that it can be solved using integer programming. 

While the integer program solution can be used for scheduling small conferences (sec-\ref{subsubsec:syn_random}, sec-\ref{subsubsec:FATREC}), they may not scale to bigger conferences due to the hardness of the objectives. 
We propose two techniques (a repeated rounding technique in sec-\ref{subsubsec:RRFS} and participant clustering in sec-\ref{subsubsec:part_clustering}) to significantly scale up the joint optimization framework.
We use data from three real conferences: \href{https://piret.gitlab.io/fatrec/}{FATREC}, \href{https://recsys.acm.org/recsys17/}{RecSys} and \href{https://icml.cc/Conferences/2017}{ICML}, respectively covering small, medium and large conference categories, and test the efficacy of our proposed method. 
Extensive evaluations over these real-world datasets (alongside synthetic datasets) show that the proposed approaches are effective in providing a nice balance between efficiency and fairness objectives (sec-\ref{sec:experiments}).
%
%
%
To our knowledge, this is the first work to consider fairness in virtual conference scheduling. We hope that our proposal will not only help conference organizers, but also spawn future research on fine-tuning solutions to specific type of virtual conferences.

\section{\bf Related Work}\label{sec:related}
We briefly review related research efforts in the following two directions: job scheduling and event scheduling.
~\\{\bf Job and Network Scheduling:}
The most commonly studied scheduling problem in computing research is job/network scheduling:
it usually has multiple agents (e.g., system processes, computing jobs, data packets, networked users or machines) who have shared access to common resource(s) (e.g., fixed number of processors, limited internet bandwidth), and the agents raise requests for using the common resource(s) from time to time;
now the goal is to allocate the resource(s) to the agents in a fair and optimal manner.
Examples include fair-share scheduling for system processes \cite{kay1988fair,li2009efficient,lozi2016linux}, 
fair sharing of network channels \cite{vaidya2005distributed}, 
fair scheduling of computing jobs on computing clusters \cite{isard2009quincy,mahajan2019themis}, 
scheduling for devices in shared wireless charging systems \cite{fang2018fair}, and 
fair scheduling of retrieval queries for databases \cite{harris2015fair}.
Our problem setup for fair conference scheduling is very different from a typical job scheduling setup. 
While conference scheduling has two types of stakeholders---participants and speakers---who have different functions and fairness requirements, job scheduling problems are usually modeled 
only for the agents who use the shared resource.
~\\{\bf Meeting/Event Scheduling:}
The problem which is closely related to the conference scheduling problem is meeting or event scheduling where there are multiple agents with different availability in different time intervals, and the goal is to find an optimal schedule for meeting(s).
Some works \cite{pino1998scheduling,chun2003optimizing} also capture participants' personal preferences for over the set of events as it can affect their participation.
We consider both the availability and preferences/ interests of the participants while modeling their satisfaction in sec-\ref{sec:prelim}.
The optimality of a schedule has been predominantly associated with its efficiency in bringing more participation \cite{garrido1996multi,capek2008event,maheswaran2004taking}.
We, too, model this as the efficiency metric which captures the total expected participation given a schedule (as in sec-\ref{subsec:tep}).
While most of the works have relied on centralized scheduling architecture, a line of works \cite{maheswaran2006privacy, wallace2005constraint,sen1998formal,freuder2001privacy,lee2016probabilistic} explore the decentralized scheduling due to privacy concerns from the participants' side.
We model the conference scheduling problem using the former one.
In meeting scheduling, utility/satisfaction is modeled only for the participants, and there has been no consideration of satisfaction from the side of the event (i.e., no concept of speakers as individuals with self interests).
In sharp contrast, the speakers in a conference also have satisfaction attached to them (sec-\ref{subsec:speaker_satis}).
Meeting scheduling has focused more on optimizing participation, fairness has received little attention in such settings (except for \citet{baum2014scheduling} dealing with a very different context).
Besides there have been very few works on conference scheduling, but they focus on in-person settings \cite{pisanski2019use,stidsen2018scheduling} and efficiency objectives \cite{atagun2020effectiveness,correia2021scheduling}.
Our mFairConf framework not only accommodates satisfaction of participants and speakers, but also cares for both efficiency and fairness in conference scheduling.

\section{Preliminaries}\label{sec:prelim}
{\bf Problem Setup:}
In a conference, let $\mathcal{P}$, $\mathcal{T}$, and $\mathcal{S}$ represent the sets of participants, planned talks, and available slots (non-overlapping) respectively;
$|\mathcal{P}|=m$, $|\mathcal{T}|=n$, $|\mathcal{S}|=l$;
let $p\in \mathcal{P}$, $t\in\mathcal{T}$, and $s\in\mathcal{S}$ be instances of participant, talk, and slot respectively.
Assuming a talk to be scheduled only once, a conference schedule $\Gamma$ is a mapping $\Gamma:\mathcal{T}\rightarrow\mathcal{S}$.
Note that, in this paper, we limit ourselves to the case with no parallel or overlapping time slots;
this implies that each slot refers to a unique time interval.
Thus, the conference schedule $\Gamma$ is a one-to-one mapping with $n\leq l$.
The goal of a conference scheduling problem is to find a schedule $\Gamma$ which satisfies some specified constraint(s) or optimizes some specified objective(s).
Further on, we use {\bf VCS} as abbreviation for {\bf V}irtual {\bf C}onference {\bf S}cheduling.
~\\\noindent {\bf Interest Scores $[V_p(\cdot)]$:}
The participants may have different preference levels over the set of talks.
We model this phenomenon using participant-specific interest scores.
Let $V(t|p)=V_p(t)$ represent $p$'s interest score for talk $t$.
Note that the interest score represents the probability of satisfaction of the participant on attending the corresponding talk;
i.e., $V_p(t)\in [0,1], \forall p\in\mathcal{P},t\in\mathcal{T}$.
~\\\noindent {\bf Ease of Availability $[A_p(\cdot)]$:}
In a virtual conference setting, the participants are located in different parts of the world which makes it convenient for them to attend talks only in specific times of the day 
(usually during the day time of their timezone).
Note that participants from same timezone may also have different ease of availability throughout the $24$-hour period.
%
Thus, we model this phenomenon using participant-specific availability scores.
Let $A(s|p)=A_p(s)$ represent the ease of availability score or the probability of $p$ making herself available in slot $s$;
i.e., $A_p(s)\in [0,1], \forall p\in\mathcal{P},s\in\mathcal{S}$. 
\subsection{Participant Satisfaction ($NCG$)}\label{subsec:participant_satis}
In virtual conferences, a participant's satisfaction depends on both her interest for the talks and her ease of availability in the time slots when the talks are scheduled.
For simplicity, we assume $V_p(\cdot)$ and $A_p(\cdot)$ to be independent of each other, i.e., the interest score of a talk does not affect the ease of availability of a participant in a slot and vice-versa.
%
%
However, the expected gain of a participant $p$ from a talk $t$ in slot $s$ will depend on the joint probability of $p$ making herself available in $s$ and getting satisfied after attending $t$: i.e, $V_p(t)\times A_p(s)$; which in turn
%
represents the probability of 
$p$ attending talk $t$ in slot $s$.
Thus, given a conference schedule $\Gamma$, we define the cumulative gain ($CG$) of participant $p$ as below.
\begin{equation}\small\label{eq:CG}
	CG(p|\Gamma)=CG_p(\Gamma)=\sum_{t\in\mathcal{T}} V_p(t)\times A_p(\Gamma(t))
\end{equation}
Now, let's imagine a situation wherein the participant $p$ is asked to choose the conference schedule $\Gamma$.
Assuming $p$ to be a selfish and rational agent, she would choose the schedule which benefits her the most; i.e., the one which gives her the highest cumulative gain.
Here, the best conference schedule for $p$ would be the one in which the talk with the highest $V_p$ is scheduled in the slot with the highest $A_p$, the talk with second highest $V_p$ is scheduled in the slot with the second highest $A_p$, and so on (this also follows from the Rearrangement inequality \cite{hardy1967inequalities});
let $\Gamma_p^*$ be that best conference schedule for $p$.
We call the cumulative gain of $p$ from schedule $\Gamma_p^*$ as her ideal cumulative gain ($ICG$):
$ICG(p)=ICG_p=\max_\Gamma CG_p(\Gamma)= CG_p(\Gamma_p^*)$.
%
%
%
We now define the overall satisfaction of a participant as her normalized cumulative gain ($NCG$) as below.
\begin{equation}\small
	NCG(p|\Gamma)=NCG_p(\Gamma)=\dfrac{CG_p(\Gamma)}{ICG_p}
\end{equation}
%
Since the denominator $ICG_p$ is the maximum possible cumulative gain for the participant $p$,
$NCG_p(\Gamma)\in [0,1]$, $\forall p\in \mathcal{P}$, $\forall \Gamma$.
\subsection{Speaker Satisfaction ($NEC$)}\label{subsec:speaker_satis}
%
We model the satisfaction of a speaker using the expected participation or crowd at her talk.
To have high speaker participation, the talk needs to be scheduled in a slot with high ease of availability of the interested participants.
Thus, given a schedule $\Gamma$, we define expected crowd ($EC$) at talk $t$ as below.
\begin{equation}\small\label{eq:EC}
	EC(t|\Gamma)=EC_t(\Gamma)=\sum_{p\in\mathcal{P}}V_p(t)\times A_p(\Gamma(t))
\end{equation}
Now if the speaker of talk $t$ is asked to prepare the conference schedule, she would try to maximize her expected crowd (assuming that similar to participants, speakers are also selfish and rational agents).
Such a schedule can be easily constructed by searching through the set of available slots to find the slot with the highest expected crowd for $t$, and then randomly allocating the remaining talks to the remaining slots.
Let that best schedule for $t$ be denoted as $\Gamma_t^*$.
We call the expected crowd at talk $t$ with schedule $\Gamma_t^*$ as the ideal expected crowd ($IEC$) of $t$:
$IEC(t)=IEC_t=\max_\Gamma EC_t(\Gamma)=EC_t(\Gamma_t^*)$.
$IEC$ represents the maximum value for expected crowd at the talk.
%
%
We now define the overall satisfaction of a speaker as the normalized expected crowd ($NEC$) at her talk as below.
\begin{equation}\small
	NEC(t|\Gamma)=NEC_t(\Gamma)=\dfrac{EC_t(\Gamma)}{IEC_t}
\end{equation}
Since the denominator $IEC_t$ is the maximum possible value of the expected crowd at talk $t$,
thus, $NEC_t(\Gamma)\in [0,1]$, $\forall t\in \mathcal{T}$, $\forall \Gamma$.
\subsection{Efficiency Objective ($TEP$)}\label{subsec:tep}
From a mechanism design perspective, a natural objective for the conference organizers is to maximize the efficiency,
i.e., the total participation in the conference --- similar to the participation metrics used in prior literature on optimal meeting scheduling \cite{garrido1996multi,capek2008event,maheswaran2004taking,pino1998scheduling,chun2003optimizing}.
Given a talk $t$ scheduled in slot $s$, the probability of participant $p$ attending it, can be written as $V_p(t)\times A_p(s)$.
Thus, the total expected participation ($TEP$) given a schedule $\Gamma$, can be written as below.
\begin{equation}\small
	TEP(\Gamma)=  \sum_{p\in \mathcal{P}} \sum_{t\in \mathcal{T}} V_p(t)\times A_p(\Gamma(t))
	\label{eq:TEP}
\end{equation}
It is worth noting that the efficiency $TEP$, here, is same as: (i) {\it the sum of cumulative gains of all the participants}; (ii) {\it the sum of expected crowd at all the talks}; both are in the form of a utilitarian social welfare function \cite{sen1982utilitarianism}.
%
%
We use $\Gamma^\text{EM}$ to represent the schedule which maximizes the efficiency;
i.e., $\Gamma^\text{EM}=\argmax_\Gamma TEP(\Gamma)$.
\begin{lemma}
	Efficiency maximization in VCS can be mapped to a min cost bipartite matching problem with polynomial time solution.
	\label{lemma:swm_max_matching}
\end{lemma}

\subsection{Fairness in Conference Scheduling}\label{sec:motivation}
%
%
\subsubsection{\bf Participant Fairness}\label{subsec:p_fairness}
%
We postulate that disparity in normalized satisfactions may cause participant unfairness, and to ensure fairness for participants, the conference schedule should equally satisfy all the participants.
However, such a hard constraint might become infeasible in real-world cases. Thus, we define a relaxed unfairness measure for participants as below.
\begin{definition}
	{\bf Participant Unfairness $\big(\Psi^\text{P}(\Gamma)\big)$:}
	The participant unfairness caused by a schedule $\Gamma$, is the maximum difference between the satisfactions of any two participants.
	\begin{equation}\small
		\Psi^\text{P}(\Gamma)=\big\{\max_{p_i\in\mathcal{P}} NCG(p_i|\Gamma)\big\}-\big\{\min_{p_j\in\mathcal{P}} NCG(p_j|\Gamma) \big\}
		\label{eq:participant_unfairness}
	\end{equation}
	\label{def:participant_unfairness}
\end{definition}
%
The {\bf fairness objective for participants} can be defined as finding the schedule $\Gamma$ which minimizes $\Psi^\text{P}(\Gamma)$.
	\begin{equation}\small
		\argmin_\Gamma \Psi^\text{P}(\Gamma)\equiv \argmin_\Gamma \bigg\{\Big\{\max_{p_i\in\mathcal{P}} NCG(p_i|\Gamma)\Big\}-\Big\{\min_{p_j\in\mathcal{P}} NCG(p_j|\Gamma)\Big\}\bigg\}
		\label{eq:participant_fairness_obj}
	\end{equation}
We now give the decision variant of the participant fairness objective which is NP-complete (as in definition-\ref{def:participant_fairness_decision} and theorem-\ref{theorem:participant_fairness_np_hard}; proof is given in the supplementary material).
\begin{definition}
	{\bf Decision variant of participant fairness:} Given $\mathcal{P}$, $\mathcal{T}$, $\mathcal{S}$, and $V_p(t)$, $A_p(s)$ $\forall p,t,s \in \mathcal{P},\mathcal{T},\mathcal{S}$, and $\epsilon\in \mathbf{R}^{\geq0}$, does there exist a schedule $\Gamma$ (a mapping from $\mathcal{T}$ to $\mathcal{S}$) such that $\Psi^P(\Gamma)\leq \epsilon$?
	\label{def:participant_fairness_decision}
\end{definition}
\begin{theorem}
	The participant fairness problem (as given in definition-\ref{def:participant_fairness_decision}) is NP-complete.
	\label{theorem:participant_fairness_np_hard}
\end{theorem}
\subsubsection{\bf Speaker Fairness}\label{subsec:s_fairness}
Similar to participant fairness, we define the unfairness measure and the fairness objective for speakers.
\begin{definition}
	{\bf Speaker Unfairness $\big(\Psi^\text{S}(\Gamma)\big)$:}
	The speaker unfairness caused by a schedule $\Gamma$, is the maximum difference between the satisfactions of any two speakers.
	\begin{equation}\small
	\Psi^\text{S}(\Gamma)=\big\{\max_{t_i\in\mathcal{T}} NEC(t_i|\Gamma)\big\}-\big\{\min_{t_j\in\mathcal{T}} NEC(t_j|\Gamma) \big\}
	\label{eq:speaker_unfairness}
	\end{equation}
	\label{def:speaker_unfairness}
\end{definition}
Now, the {\bf fairness objective for speakers} can be defined as finding the schedule $\Gamma$ which minimizes $\Psi^\text{S}(\Gamma)$.
\begin{equation}\small
\argmin_\Gamma \Psi^\text{S}(\Gamma)\equiv \argmin_\Gamma \bigg\{\Big\{\max_{t_i\in\mathcal{T}} NEC(t_i|\Gamma)\Big\}-\Big\{\min_{t_j\in\mathcal{T}} NEC(t_j|\Gamma)\Big\}\bigg\}
\label{eq:speaker_fairness_obj}
\end{equation}
\section{Balancing Efficiency and Fairness}
\label{sec:methodology}
In VCS (as defined in sec-\ref{sec:prelim}), the ultimate goal is to find a schedule $\Gamma$ that optimizes efficiency while minimizing participant unfairness and speaker unfairness (as defined in eq-\ref{eq:TEP}, eq-\ref{eq:participant_fairness_obj}, eq-\ref{eq:speaker_fairness_obj} respectively).
%
%
However, the individual objectives may be at loggerheads with each other, and simultaneous optimization of the three objectives
may not be possible. 
Thus, first in sec-\ref{subsec:tensions}, we highlight the potential conflicts in simultaneously ensuring efficiency and fairness.
Then, in sec-\ref{subsec:joint_opt}, we propose a joint optimization framework for the problem, and in sec-\ref{subsec:scale_up}, we propose two techniques to scale up the joint optimization for scheduling bigger conferences.
\subsection{Tension between Efficiency and Fairness}\label{subsec:tensions}
Through a set of claims, 
we illustrate some fundamental tensions between efficiency and fairness in VCS. Proofs are in the appendix.
\begin{claim}
	In VCS, it is not always possible to gain participant fairness without losing efficiency.
	\label{claim:1}
\end{claim}
%
%
%
\begin{claim}
	In VCS, it is not always possible to gain speaker fairness without losing efficiency.
	\label{claim:2}
\end{claim}
\begin{claim}
	In VCS, it is not always possible to gain speaker fairness without losing participant fairness and vice-versa.
	\label{claim:3}
\end{claim}
%
%
From the above three claims, we can say that improving on one of the three objectives might cause losses in the other two objectives. 
%
Thus, while designing a conference schedule, a suitable balance among the two fairness objectives and the efficiency objective should be maintained. 
Hence to attain such a balance, next, we propose a framework to jointly optimize these objectives.
\subsection{Joint Optimization for Efficiency and Fairness (mFairConf)}\label{subsec:joint_opt}
We propose {\bf m(ultistakeholder)\bf FairConf}, a joint optimization framework  which combines fairness objectives with efficiency.
{\small\begin{multline}
\argmax_\Gamma \frac{TEP(\Gamma)}{mn}+ \lambda_1\times \bigg\{\Big\{\min_{p_j\in\mathcal{P}} NCG(p_j|\Gamma)\Big\}-\Big\{\max_{p_i\in\mathcal{P}} NCG(p_i|\Gamma)\Big\}\bigg\}\\
+\lambda_2\times \bigg\{\Big\{\min_{t_j\in\mathcal{T}} NEC(t_j|\Gamma)\Big\}-\Big\{\max_{t_i\in\mathcal{T}} NEC(t_i|\Gamma)\Big\}\bigg\}
\label{eq:joint_opt}
\end{multline}}
Here we normalize the efficiency objective to bring all the three components to similar scales;
i.e., $TEP$ is divided by $|\mathcal{P}|\cdot|\mathcal{T}|=mn$ (it is the maximum possible value for $TEP$---occurs when $V_p(t)=A_p(s)=1$, $\forall p,t,s$).
We also reverse the fairness objective functions from eq-\ref{eq:participant_fairness_obj} and eq-\ref{eq:speaker_fairness_obj} while inserting them in eq-\ref{eq:joint_opt} as it features $\argmax$ instead of $\argmin$, and use $\lambda_1,\lambda_2$ as weights for participant fairness and speaker fairness respectively.

We take a matrix $X$ of dimensions $|\mathcal{T}|\times|\mathcal{S}|$.
Each element of $X$: $X_{t,s}$ is a binary indicator variable for talk $t\in \mathcal{T}$ being scheduled in slot $s\in \mathcal{S}$, i.e., $X_{t,s}=1$ if $t$ is scheduled in $s$ and $0$ otherwise.
Now to operationalize the joint optimization objective in eq-\ref{eq:joint_opt}, we express it as an integer program in eq-\ref{eq:ip}.
The first constraint is the integrality constraint.
Second constraint ensures that, each talk gets scheduled exactly once.
On the other hand, one slot can be allocated to atmost one talk which is ensured by the third constraint.

{\small\begin{multline}
\argmax_X \text{   }\space \frac{1}{mn}\sum\limits_{t\in\mathcal{T}}\sum\limits_{p\in\mathcal{P}} \sum\limits_{s\in\mathcal{S}}V_p(t)\cdot A_p(s)\cdot X_{t,s}\\ 
+\lambda_1\Bigg[\min_{p_j\in \mathcal{P}} \sum_{t\in\mathcal{T}}\sum_{s\in\mathcal{S}}\frac{V_{p_j}(t)A_{p_j}(s)}{ICG(p_j)}X_{t,s}-\max_{p_i\in \mathcal{P}}\sum_{t\in\mathcal{T}}\sum_{s\in\mathcal{S}}\frac{V_{p_i}(t)A_{p_i}(s)}{ICG(p_i)} X_{t,s}\Bigg]\\
+\lambda_2\Bigg[\min_{t_j\in \mathcal{T}} \sum_{p\in\mathcal{P}}\sum_{s\in\mathcal{S}}\frac{V_{p}(t_j)A_{p}(s)}{IEC(t_j)} X_{t_j,s}-\max_{t_i\in \mathcal{T}}\sum_{p\in\mathcal{P}}\sum_{s\in\mathcal{S}}\frac{V_{p}(t_i) A_{p}(s)}{IEC(t_i)} X_{t_i,s}\Bigg]\\
\text{s.t.   }X_{t,s}\in \{0,1\} \text{   }\forall t\in\mathcal{T}, s\in\mathcal{S}\\
\sum_{s\in \mathcal{S}}X_{t,s}=1,\text{   }\forall t\in\mathcal{T}\\
\sum_{t\in \mathcal{T}}X_{t,s}\leq1,\text{   }\forall s\in\mathcal{S}\\
\label{eq:ip}
\end{multline}}
%
%
%
\subsection{Scaling Up the Joint Optimization}\label{subsec:scale_up}
As the joint objective is NP-hard (theorem-\ref{theorem:participant_fairness_np_hard}), 
scaling the integer program (as given in eq-\ref{eq:ip}) to big conferences with large number of talks and participants, would need huge computing resources.
Thus, we provide a rounding heuristic (in sec-\ref{subsubsec:RRFS}), and a clustering approach (in sec-\ref{subsubsec:part_clustering}) which can significantly reduce the time and computing resources needed for fair VCS. 
%
\subsubsection{\bf Repeated Rounding of Fractional Solutions (RRFS)}\label{subsubsec:RRFS}
We propose a repeated rounding heuristic to approximately solve the joint optimization in eq-\ref{eq:ip}. 
We first ignore the integrality constraint ($X_{t,s}\in \{0,1\}$ in eq-\ref{eq:ip}), and replace it with a general non-negativity constraint ($X_{t,s}\geq0$) to get a fractional solution for decision variable $X$.
Note that even though finding an integer solution is NP-hard, finding a fractional solution is polynomial time solvable ($\mathcal{O}(l^3)$, where $l$ is number of slots)~\cite{kuhn1955hungarian}.
%
%
We, then, find the maximum element from the fractional solution (let it be $X_{t,s}$), schedule talk $t$ in slot $s$, and replace all other elements in the same row as $t$ and column as $s$ with $0$. 
Then, we move to the next maximum from the remaining elements and repeat the same process till 
every other element becomes zero or all talks are scheduled. 
If 
some talks remain unscheduled when there is no non-zero element, we filter the unscheduled talks and slots, and repeat the same process of finding a fractional solution and rounding.
This method is detailed in alg-\ref{alg:RRFS} in the appendix, and has the worst case time complexity $\mathcal{O}\big(nl(l^2+mn)\big)$ when $l\leq2^{n^2}$.
\subsubsection{\bf Participant Clustering (PC)}\label{subsubsec:part_clustering}
In big conferences, although the number of talks and slots stay limited or may not grow too much, the number of participants could become very high leading to high memory complexity even to get a fractional solution.
%
Thus, for such cases, we propose to group similar participants into clusters as a pre-processing strategy.
%
We concatenate the interest and availability scores of a participant to create the participant's profile vector;
i.e., participant $p$'s profile is $[V_p(t)\forall t\in \mathcal{T}: A_p(s)\forall s\in \mathcal{S}]$.
We, then, apply $k$-means clustering to group similar participants, and use the cluster centroids as the representative participant profiles in the mFairConf while adding multiplicative weights---same as the size of the corresponding clusters---only in the efficiency and speaker fairness terms of the objective (as these two depend on the true audience participation values).
%
\begin{figure*}[t!]
	\center{\small
		\subfloat[{\bf Participant Unfairness}]{\pgfplotsset{width=0.19\textwidth,height=0.17\textwidth,compat=1.9}\begin{tikzpicture}
\begin{axis}[
xlabel={$\lambda_1$},
ylabel={$NCG_\text{max}-NCG_\text{min}$},
xmin=0, xmax=1.0,
ymin=0, ymax=1,
xtick={0,0.2,0.4,0.6,0.8,1.0},
ytick={0,0.2,0.4,0.6,0.8,1},
xticklabel style={rotate=60},
ymajorgrids=true,
grid style=dashed,
]
\addplot[
color=blue,
mark=o,
mark size=2pt
]
coordinates {(0.0, 0.32293847449001034)
	(0.1, 0.2729294223204688)
	(0.2, 0.12985043028060617)
	(0.3, 0.12985043028060617)
	(0.4, 0.0676991644619761)
	(0.5, 0.0676991644619761)
	(0.6, 0.0676991644619761)
	(0.7, 0.0676991644619761)
	(0.8, 0.0676991644619761)
	(0.9, 0.0676991644619761)
	(1.0, 0.0676991644619761)
};

\addplot[color=red,ultra thick]{0.22317261030063873};

\addplot[color=teal,ultra thick]{0.03923540977530626};

\addplot[color=orange,ultra thick]{0.43249203217403653};

\addplot[dashed,color=black,ultra thick]{0.47076743475420185};
\end{axis}
\end{tikzpicture}\label{fig:syn_part_unfair_lam1}}
		\hfil
		\subfloat[{\bf Speaker Unfairness}]{\pgfplotsset{width=0.19\textwidth,height=0.17\textwidth,compat=1.9}\begin{tikzpicture}
\begin{axis}[
xlabel={$\lambda_1$},
ylabel={$NEC_\text{max}-NEC_\text{min}$},
xmin=0, xmax=1.0,
ymin=0, ymax=1,
xtick={0,0.2,0.4,0.6,0.8,1.0},
ytick={0,0.2,0.4,0.6,0.8,1},
xticklabel style={rotate=60},
ymajorgrids=true,
grid style=dashed,
]
\addplot[
color=blue,
mark=o,
mark size=2pt
]
coordinates {(0.0, 0.13795849819887518)
	(0.1, 0.13795849819887518)
	(0.2, 0.20708185976376725)
	(0.3, 0.20708185976376725)
	(0.4, 0.28576390240565164)
	(0.5, 0.28576390240565164)
	(0.6, 0.28576390240565164)
	(0.7, 0.28576390240565164)
	(0.8, 0.28576390240565164)
	(0.9, 0.28576390240565164)
	(1.0, 0.28576390240565164)
	
};
\addplot[color=red,ultra thick]{0.3856090837304965};

\addplot[color=teal,ultra thick]{0.37896654797115414};

\addplot[color=orange,ultra thick]{0.13795849819887518};

\addplot[dashed,color=black,ultra thick]{0.5420674384953205};
\end{axis}
\end{tikzpicture}\label{fig:syn_spea_unfair_lam1}}
		\hfil
		\subfloat[{\bf Schedule Efficiency}]{\pgfplotsset{width=0.19\textwidth,height=0.17\textwidth,compat=1.9}
			\begin{tikzpicture}
\begin{axis}[
xlabel={$\lambda_1$},
ylabel={$TEP$},
xmin=0, xmax=1.0,
ymin=20, ymax=35,
xtick={0,0.2,0.4,0.6,0.8,1.0},
ytick={20,25,30,35},
xticklabel style={rotate=60},
ymajorgrids=true,
grid style=dashed,
]
\addplot[
color=blue,
mark=o,
mark size=2pt
]
coordinates {(0.0, 25.453684190392377)
	(0.1, 25.35731941350786)
	(0.2, 26.85809686867394)
	(0.3, 26.85809686867394)
	(0.4, 28.69110533448474)
	(0.5, 28.69110533448474)
	(0.6, 28.69110533448474)
	(0.7, 28.69110533448474)
	(0.8, 28.69110533448474)
	(0.9, 28.69110533448474)
	(1.0, 28.69110533448474)

};
\addplot[color=red,ultra thick]{30.224415486673877};

\addplot[color=teal,ultra thick]{27.80547321591618};

\addplot[color=orange,ultra thick]{24.594540770528468};

\addplot[dashed,color=black,ultra thick]{25.716271790133003};
\end{axis}
\end{tikzpicture}\label{fig:syn_TEP_lam1}}		
		\hfil
		\subfloat[{\bf Mean Part. Satisfaction}]{\pgfplotsset{width=0.19\textwidth,height=0.17\textwidth,compat=1.9}
			\begin{tikzpicture}
\begin{axis}[
xlabel={$\lambda_1$},
ylabel={$NCG_\text{mean}$},
xmin=0, xmax=1.0,
ymin=0.7, ymax=0.9,
xtick={0,0.2,0.4,0.6,0.8,1.0},
ytick={0.7,0.8,0.9},
xticklabel style={rotate=60},
ymajorgrids=true,
grid style=dashed,
]
\addplot[
color=blue,
mark=o,
mark size=2pt
]
coordinates {(0.0, 0.7464210865537194)
	(0.1, 0.7403768692101866)
	(0.2, 0.795187458515205)
	(0.3, 0.795187458515205)
	(0.4, 0.8486274897554829)
	(0.5, 0.8486274897554829)
	(0.6, 0.8486274897554829)
	(0.7, 0.8486274897554829)
	(0.8, 0.8486274897554829)
	(0.9, 0.8486274897554829)
	(1.0, 0.8486274897554829)

};

\addplot[color=red,ultra thick]{0.8902282660335098};

\addplot[color=teal,ultra thick]{0.8244530868773643};

\addplot[color=orange,ultra thick]{0.7167104916132297};

\addplot[dashed,color=black,ultra thick]{0.7514033257873154};
\end{axis}
\end{tikzpicture}\label{fig:syn_part_mu_lam1}}
		\hfil
		\subfloat[{\bf Mean Spea. Satisfaction}]{\pgfplotsset{width=0.19\textwidth,height=0.17\textwidth,compat=1.9}
			\begin{tikzpicture}
\begin{axis}[
xlabel={$\lambda_1$},
ylabel={$NEC_\text{mean}$},
xmin=0, xmax=1.0,
ymin=0.6, ymax=0.9,
xtick={0,0.2,0.4,0.6,0.8,1.0},
ytick={0.6,0.7,0.8,0.9},
xticklabel style={rotate=60},
ymajorgrids=true,
grid style=dashed,
]
\addplot[
color=blue,
mark=o,
mark size=2pt
]
coordinates {(0.0, 0.7255026836751933)
	(0.1, 0.7226051292440007)
	(0.2, 0.7681548764443848)
	(0.3, 0.7681548764443848)
	(0.4, 0.8195625500306034)
	(0.5, 0.8195625500306034)
	(0.6, 0.8195625500306034)
	(0.7, 0.8195625500306034)
	(0.8, 0.8195625500306034)
	(0.9, 0.8195625500306034)
	(1.0, 0.8195625500306034)
};

\addplot[color=red,ultra thick]{0.8577998948539282};

\addplot[color=teal,ultra thick]{0.8062473008059863};

\addplot[color=orange,ultra thick]{0.7014558579515918};

\addplot[dashed,color=black,ultra thick]{0.7179761110326732};
\end{axis}
\end{tikzpicture}\label{fig:syn_spea_mu_lam1}}
		\hfil
		\subfloat[{\bf Participant Unfairness}]{\pgfplotsset{width=0.19\textwidth,height=0.17\textwidth,compat=1.9}\begin{tikzpicture}
\begin{axis}[
xlabel={$\lambda_2$},
ylabel={$NCG_\text{max}-NCG_\text{min}$},
xmin=0, xmax=1.0,
ymin=0, ymax=1,
xtick={0,0.2,0.4,0.6,0.8,1.0},
ytick={0,0.2,0.4,0.6,0.8,1},
xticklabel style={rotate=60},
ymajorgrids=true,
grid style=dashed,
]
\addplot[
color=blue,
mark=o,
mark size=2pt
]
coordinates {(0.0, 0.06402099792168636)
	(0.1, 0.0676991644619761)
	(0.2, 0.0676991644619761)
	(0.3, 0.0676991644619761)
	(0.4, 0.0676991644619761)
	(0.5, 0.0676991644619761)
	(0.6, 0.0676991644619761)
	(0.7, 0.12985043028060617)
	(0.8, 0.12985043028060617)
	(0.9, 0.12985043028060617)
	(1.0, 0.12985043028060617)
		
};
\addplot[color=red,ultra thick]{0.22317261030063873};

\addplot[color=teal,ultra thick]{0.03923540977530626};

\addplot[color=orange,ultra thick]{0.43249203217403653};

\addplot[dashed,color=black,ultra thick]{0.47076743475420185};
\end{axis}
\end{tikzpicture}\label{fig:syn_part_unfair_lam2}}
		\hfil
		\subfloat[{\bf Speaker Unfairness}]{\pgfplotsset{width=0.19\textwidth,height=0.17\textwidth,compat=1.9}\begin{tikzpicture}
\begin{axis}[
xlabel={$\lambda_2$},
ylabel={$NEC_\text{max}-NEC_\text{min}$},
xmin=0, xmax=1.0,
ymin=0, ymax=1,
xtick={0,0.2,0.4,0.6,0.8,1.0},
ytick={0,0.2,0.4,0.6,0.8,1},
xticklabel style={rotate=60},
ymajorgrids=true,
grid style=dashed,
]
\addplot[
color=blue,
mark=o,
mark size=2pt
]
coordinates {(0.0, 0.37896654797115414)
	(0.1, 0.28576390240565164)
	(0.2, 0.28576390240565164)
	(0.3, 0.28576390240565164)
	(0.4, 0.28576390240565164)
	(0.5, 0.28576390240565164)
	(0.6, 0.28576390240565164)
	(0.7, 0.20708185976376725)
	(0.8, 0.20708185976376725)
	(0.9, 0.20708185976376725)
	(1.0, 0.20708185976376725)
		
};

\addplot[color=red,ultra thick]{0.3856090837304965};

\addplot[color=teal,ultra thick]{0.37896654797115414};

\addplot[color=orange,ultra thick]{0.13795849819887518};

\addplot[dashed,color=black,ultra thick]{0.5420674384953205};
\end{axis}
\end{tikzpicture}\label{fig:syn_spea_unfair_lam2}}
		\hfil
		\subfloat[{\bf Schedule Efficiency}]{\pgfplotsset{width=0.19\textwidth,height=0.17\textwidth,compat=1.9}
			\begin{tikzpicture}
\begin{axis}[
xlabel={$\lambda_2$},
ylabel={$TEP$},
xmin=0, xmax=1.0,
ymin=20, ymax=35,
xtick={0,0.2,0.4,0.6,0.8,1.0},
ytick={20,25,30,35},
xticklabel style={rotate=60},
ymajorgrids=true,
grid style=dashed,
]
\addplot[
color=blue,
mark=o,
mark size=2pt
]
coordinates {(0.0, 29.27073278174233)
	(0.1, 28.69110533448474)
	(0.2, 28.69110533448474)
	(0.3, 28.69110533448474)
	(0.4, 28.69110533448474)
	(0.5, 28.69110533448474)
	(0.6, 28.69110533448474)
	(0.7, 26.85809686867394)
	(0.8, 26.85809686867394)
	(0.9, 26.85809686867394)
	(1.0, 26.85809686867394)
		
};
\addplot[color=red,ultra thick]{30.224415486673877};

\addplot[color=teal,ultra thick]{27.80547321591618};

\addplot[color=orange,ultra thick]{24.594540770528468};

\addplot[dashed,color=black,ultra thick]{25.716271790133003};
\end{axis}
\end{tikzpicture}\label{fig:syn_TEP_lam2}}	
		\hfil
		\subfloat[{\bf Mean Part. Satisfaction}]{\pgfplotsset{width=0.19\textwidth,height=0.17\textwidth,compat=1.9}
			\begin{tikzpicture}
\begin{axis}[
xlabel={$\lambda_2$},
ylabel={$NCG_\text{mean}$},
xmin=0, xmax=1.0,
ymin=0.7, ymax=0.9,
xtick={0,0.2,0.4,0.6,0.8,1.0},
ytick={0.7,0.8,0.9},
xticklabel style={rotate=60},
ymajorgrids=true,
grid style=dashed,
]
\addplot[
color=blue,
mark=o,
mark size=2pt
]
coordinates {(0.0, 0.8654720561668746)
	(0.1, 0.8486274897554829)
	(0.2, 0.8486274897554829)
	(0.3, 0.8486274897554829)
	(0.4, 0.8486274897554829)
	(0.5, 0.8486274897554829)
	(0.6, 0.8486274897554829)
	(0.7, 0.795187458515205)
	(0.8, 0.795187458515205)
	(0.9, 0.795187458515205)
	(1.0, 0.795187458515205)
		
};
\addplot[color=red,ultra thick]{0.8902282660335098};

\addplot[color=teal,ultra thick]{0.8244530868773643};

\addplot[color=orange,ultra thick]{0.7167104916132297};

\addplot[dashed,color=black,ultra thick]{0.7514033257873154};
\end{axis}
\end{tikzpicture}\label{fig:syn_part_mu_lam2}}
		\hfil
		\subfloat[{\bf Mean Spea. Satisfaction}]{\pgfplotsset{width=0.19\textwidth,height=0.17\textwidth,compat=1.9}
			\begin{tikzpicture}
\begin{axis}[
xlabel={$\lambda_2$},
ylabel={$NEC_\text{mean}$},
xmin=0, xmax=1.0,
ymin=0.6, ymax=0.9,
xtick={0,0.2,0.4,0.6,0.8,1.0},
ytick={0.6,0.7,0.8,0.9},
xticklabel style={rotate=60},
ymajorgrids=true,
grid style=dashed,
]
\addplot[
color=blue,
mark=o,
mark size=2pt
]
coordinates {(0.0, 0.8344027343837741)
	(0.1, 0.8195625500306034)
	(0.2, 0.8195625500306034)
	(0.3, 0.8195625500306034)
	(0.4, 0.8195625500306034)
	(0.5, 0.8195625500306034)
	(0.6, 0.8195625500306034)
	(0.7, 0.7681548764443848)
	(0.8, 0.7681548764443848)
	(0.9, 0.7681548764443848)
	(1.0, 0.7681548764443848)
		
};
\addplot[color=red,ultra thick]{0.8577998948539282};

\addplot[color=teal,ultra thick]{0.8062473008059863};

\addplot[color=orange,ultra thick]{0.7014558579515918};

\addplot[dashed,color=black,ultra thick]{0.7179761110326732};
\end{axis}
\end{tikzpicture}\label{fig:syn_spea_mu_lam2}}
		\vfil
		\subfloat{\pgfplotsset{width=.7\textwidth,compat=1.9}
			\begin{tikzpicture}
			\begin{customlegend}[legend entries={{\bf EM},{\bf PFair},{\bf SFair},{\bf IAM}, {\bf mFairConf}},legend columns=5,legend style={/tikz/every even column/.append style={column sep=0.5cm}}]
			\addlegendimage{red,mark=.,ultra thick,sharp plot}
			\addlegendimage{teal,mark=.,ultra thick,sharp plot}
			\addlegendimage{orange,mark=.,ultra thick,sharp plot}
			\addlegendimage{dashed,black,mark=.,ultra thick,sharp plot}
			\addlegendimage{blue,mark=o,sharp plot}
			\end{customlegend}
			\end{tikzpicture}}
	}
	\caption{Results on synthetic dataset. For the plots in the first row, $\lambda_2$ is fixed at $0.5$, and $\lambda_1$ is varied. For the plots in second row, $\lambda_1$ is fixed at $0.5$, and $\lambda_2$ is varied.}\label{fig:syn_random}
\end{figure*}
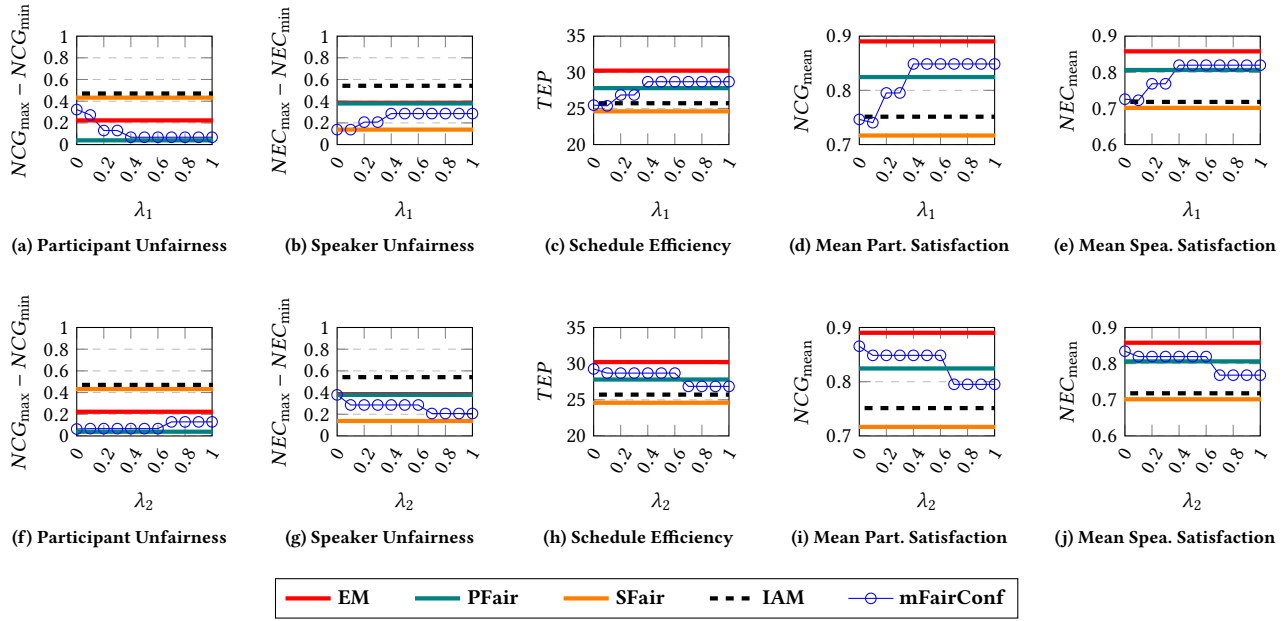
\section{Experimental Evaluation}
\label{sec:experiments}
%
\subsection{Experimental Setup}
\subsubsection{\bf Datasets:}\label{subsec:datasets}
We use real-world and synthetic datasets.

\noindent {\bf Real-world Datasets:}
We consider three computer science conferences: Workshop on Responsible Recommendation (FATREC), ACM conference on Recommender Systems (RecSys), and International Conference on Machine Learning (ICML).
Note that these conferences fall in small, medium and large conference categories respectively, and they help us to evaluate not only mFairConf, but also the proposed scaling up approaches.
While we gather data on participant availabilities and timezones from publicly available sources (released by the organizers of $2020$ conference), true participant interests are not available.
Thus, we consider the list of talks (published papers) from $2017$ edition of the conferences as the talk set, and then rely on randomized signals from the total number of citations or views of the papers to sample participant interest scores.
Next, we describe each dataset in more details. 

\textit{\textbf{i. RECSYS:}}
We collect the participant data (timezone data of $1112$ participants) released in the welcome note of $2020$ edition of the conference.
We consider each participant to be completely available only during their working time (i.e., $\forall p$, $A_p(s)=1$ if $s$ is in between $9$AM to $5$PM in local timezone) of a day, and otherwise not available at all.
We gather the published papers ($26$ papers) in the $2017$ edition (\url{https://recsys.acm.org/recsys17/}) along with their citation counts which we use as a proxy for overall participant interest for a talk (i.e., a paper).
For each participant-talk pair, we sample the interest scores from the a Bernoulli distribution $V_p(t)\sim Bernoulli\big(prob=\frac{\#cites(t)}{\max_{t'\in\mathcal{T}} \#cites(t')}\big), \forall p,t$.
We consider $48$ half hour slots over a $24$ hour period (starting from $00$hours in UTC), and try to schedule the $26$ talks.

\textit{\textbf{ii. FATREC:}}
Here, we consider the workshop on responsible recommendation which is organized in conjunction with the RecSys conference.
We collect the accepted papers ($11$ papers) from $2017$ edition of the workshop (\url{https://piret.gitlab.io/fatrec/}) along with their citation counts.
We use the citation counts to sample the participants' interest scores just as we do in case of RECSYS, and separately also from a normal distribution $V_p(t)\sim Normal\big(mean=\frac{\#cites(t)}{\max_{t'\in\mathcal{T}} \#cites(t')},std=\frac{\text{mean}}{4}\big), \forall p,t$; however due to space constraints we present results only on the dataset derived from Normal distribution.
For availability scores, we downsample $40$ participants from the RECSYS participants set.
We consider $96$ fifteen-minute-slots (slot-size is same as in the $2017$ schedule) over a $24$ hour period (starting from $00$hours in UTC), and try to schedule the $11$ talks.

\textit{\textbf{iii. ICML:}}
We use the publicly available survey responses \cite{icml2020virtual} from the participants of the International Conference on Machine Learning for the participant availability scores.
We keep the number of participants same as the number of respondents in the survey, which is $2722$.
We, then, gather all papers from the $2017$ edition of the conference (\url{https://icml.cc/Conferences/2017}) which are listed in ACM digital library \cite{10.5555/3305381} along with the number of times each paper is downloaded.
Similar to RECSYS, for each participant-talk pair, we sample the interest score from a Bernoulli distribution. 
For the $209$ listed papers, we consider $240$ half hour slots over a period of $5$ days (starting from $00$hours in UTC), and schedule the talks.

\noindent {\bf Synthetic Dataset:}
We synthesize a dataset typically mimicking a small conference, and use it to understand and illustrate the dynamics of different methods.
We take number of participants ($\abs{\mathcal{P}}) = 10$, number of talks ($\abs{\mathcal{T}}) = 10$, number of slots ($\abs{\mathcal{S}}) = 10$, and generate the synthetic dataset where the slots represent non-overlapping equal-sized time intervals
available for scheduling.
The interest scores and availability scores are then sampled from a uniform random distribution in $[0,1]$; i.e., $V_p(t)\sim \text{Uniform}([0,1])$ and $A_p(s)\sim \text{Uniform}([0,1])$, $\forall p,t,s$.
\subsubsection{\bf Baselines:}\label{subsec:baselines}
We use the following baselines and empirically compare them with our approach mFairConf from sec-\ref{subsec:joint_opt}. 
~\\\textbf{\textit{A.}} Efficiency Maximization ({\bf EM}): 
Following the trend in prior works on optimal meeting scheduling \cite{garrido1996multi,capek2008event,maheswaran2004taking,pino1998scheduling,chun2003optimizing}, here, we just optimize the schedule for efficiency; i.e., $\Gamma^\text{EM}$ or $\argmax_\Gamma TEP(\Gamma)$ without any fairness consideration.
~\\\textbf{\textit{B.}} Participant Fairness Maximization ({\bf PFair}):
Here, we just optimize for participant fairness; i.e., minimize participant unfairness [$\argmin_\Gamma \Psi^\text{P}(\Gamma)$] as defined in eq-\ref{eq:participant_fairness_obj}.
~\\\textbf{\textit{C.}} Speaker Fairness Maximization ({\bf SFair}):
Here, we just optimize for speaker fairness; i.e., minimize speaker unfairness [$\argmin_\Gamma \Psi^\text{S}(\Gamma)$] as defined in eq-\ref{eq:speaker_fairness_obj}.
~\\\textbf{\textit{D.}} Interest-Availability Matching ({\bf IAM}): 
Here, we sort the talks in descending order of the overall interest scores received by them, i.e., $\sum_{p\in \mathcal{P}}V_p(t)$, and the slots in descending order of the overall availability scores received by them, i.e., $\sum_{p\in \mathcal{P}}A_p(s)$.
Now, we assign the talk with the highest overall interest score to the slot with the highest overall availability score, the talk with the second highest overall interest score to the slot with the second highest overall availability score, and so on (with random tie-breaks).
IAM is one of the naive alternatives when scheduling is done manually (as natural objectives like EM usually need computing resources).
It is also worth noting that, in the usual physical conference settings (i.e., all participants have identical ease of availability over all available slots $A_p(s)=A(s)$, $\forall s\in \mathcal{S},p\in \mathcal{P}$), {\bf IAM} yields a conference schedule which maximizes efficiency. It is a special case of lemma-\ref{lemma:V_or_A_same} (refer to case (a) of lemma-\ref{lemma:V_or_A_same} for the proof).
\begin{lemma}{\bf IAM} maximizes efficiency, if the participants are identical either in terms of their interests in the talks or in terms of their ease of availability over the available slots, or both.\label{lemma:V_or_A_same}
\end{lemma}
%
%
%
\subsubsection{\bf Evaluation Metrics:}\label{subsec:metrics}
Apart from the {\bf fairness} metrics ($\Psi^\text{P}(\Gamma)=NCG_\text{max}-NCG_\text{min}$ as in definition-\ref{def:participant_unfairness}, $\Psi^\text{S}(\Gamma)=NEC_\text{max}-NEC_\text{min}$ as in definition-\ref{def:speaker_unfairness}), we also measure the mean satisfaction of participants and speakers ($NCG_\text{mean}$ and $NEC_\text{mean}$), and {\bf efficiency} ($TEP$ as in sec-\ref{subsec:tep}) as indicators of {\bf efficiency}.
Even though our unfairness metrics ($\Psi^\text{P}(\Gamma)$, $\Psi^\text{S}(\Gamma)$) are based on max-min differences for simplicity in modeling, they are quite vulnerable to participants or speakers with niche profiles especially in big conferences. 
Thus, we also use gini index \cite{gini1912variabilita} to measure the overall inequality in individual participant and speaker satisfactions ($NCG_\text{gini}$) as a measure of overall unfairness in bigger datasets (RECSYS and ICML).
We use cvxpy (\url{https://www.cvxpy.org/}) paired with Gurobi (\url{https://www.gurobi.com/}) solver for the optimization.
System details are: Debian GNU/Linux 10 on AMD64 architecture, Python 2.7.16, cvxpy 1.0.21, gurobipy 9.1.1, numpy 1.16.2, scikit-learn 0.20.3.
\subsection{Experimental Results}\label{subsec:syn_experiments}
%
\subsubsection{\bf Results on the Synthetic Dataset:}\label{subsubsec:syn_random}
We plot the results for synthetic dataset in fig-\ref{fig:syn_random}.
Note that, unlike mFairConf, the baseline approaches do not have hyperparameters $\lambda_1,\lambda_2$;
thus, baseline results are just horizontal straight lines while mFairConf's results vary with hyperparameter settings.
~\\{\bf Baseline Results:}
EM achieves the highest expected participation $TEP$ (by definition it should), the highest mean participant satisfaction and mean speaker satisfaction (refer figs-\ref{fig:syn_part_mu_lam1},\ref{fig:syn_spea_mu_lam1},\ref{fig:syn_TEP_lam1}) while performing poorly on participant and speaker fairness (figs-\ref{fig:syn_part_unfair_lam1},\ref{fig:syn_spea_unfair_lam1}).
On the other hand, the naive IAM performs poorly in all the metrics.
As PFair optimizes only for participant fairness, it has the highest participant fairness (least unfairness in fig-\ref{fig:syn_part_unfair_lam1}) while losing in all the other metrics.
Similarly, SFair performs the best in speaker fairness (least unfairness in fig-\ref{fig:syn_spea_unfair_lam1}) while losing in all other metrics as SFair optimizes only for speaker fairness.
~\\{\bf mFairConf Results:}
Note that, for the plots in first row (fig-\ref{fig:syn_part_unfair_lam1} to fig-\ref{fig:syn_spea_mu_lam1}), we fix $\lambda_2=0.5$ and vary $\lambda_1$ from $0$ to $1$;
for the plots in second row (figs \ref{fig:syn_part_unfair_lam2} to \ref{fig:syn_spea_mu_lam2}), we fix $\lambda_1=0.5$ and vary $\lambda_2$ from $0$ to $1$.
The general trends observed in mFairConf's results are: with increase in the weight for participant fairness ($\lambda_1$), mFairConf achieves better participant fairness (fig-\ref{fig:syn_part_unfair_lam1}) but worse speaker fairness (fig-\ref{fig:syn_spea_unfair_lam1});
with increase in the weight for speaker fairness ($\lambda_2$), mFairConf achieves better speaker fairness (fig-\ref{fig:syn_spea_unfair_lam2}) but worse participant fairness (increase in participant unfairness in fig-\ref{fig:syn_part_unfair_lam2}).
We find that mFairConf with the setting of $\lambda_1=\lambda_2=0.5$ gives a balanced performance across all the metrics;
it performs good in both participant fairness (very small unfairness in fig-\ref{fig:syn_part_unfair_lam1}-- close to PFair) and speaker fairness (very small unfairness in fig-\ref{fig:syn_spea_unfair_lam1}-- slightly higher than SFair) while causing only marginal losses in mean participant satisfaction (fig-\ref{fig:syn_part_mu_lam1}), mean speaker satisfaction (fig-\ref{fig:syn_spea_mu_lam1}), and the efficiency (fig-\ref{fig:syn_TEP_lam1}).
Note that, here, we do not imply that $\lambda_1=\lambda_2=0.5$ will always give a balanced performance from mFairConf; instead such a hyperparameter setting will be dataset-specific, and one needs to find it out through exploration.
Moreover, a conference organizer may not always want a fully balanced schedule; she may even set the hyperparameters as per her relative priorities towards participants and speakers.
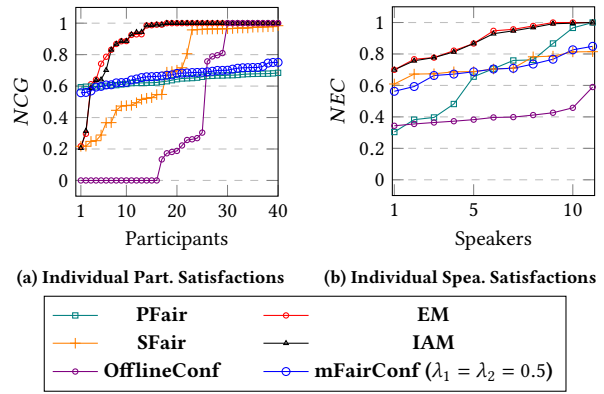
\begin{figure}
	\center{\small
		\subfloat[{\bf Individual Part. Satisfactions}]{\pgfplotsset{width=0.24\textwidth,height=0.23\textwidth,compat=1.9}
			\begin{tikzpicture}
\begin{axis}[
xlabel={Participants},
ylabel={$NCG$},
xmin=0, xmax=40,
ymin=-0.1, ymax=1.1,
xtick={1,10,20,30,40},
ytick={0,0.2,0.4,0.6,0.8,1},
ymajorgrids=true,
grid style=dashed
]

\addplot[
color=red,
mark=o,
mark size=1pt
]
coordinates {(1, 0.21808505)
	(2, 0.29648114)
	(3, 0.61189287)
	(4, 0.6412088)
	(5, 0.74206144)
	(6, 0.78643866)
	(7, 0.8309141)
	(8, 0.8677253)
	(9, 0.88497346)
	(10, 0.88618872)
	(11, 0.92515673)
	(12, 0.92903683)
	(13, 0.92950924)
	(14, 0.9853238)
	(15, 0.98696513)
	(16, 0.990296)
	(17, 0.99047311)
	(18, 1.0)
	(19, 1.0)
	(20, 1.0)
	(21, 1.0)
	(22, 1.0)
	(23, 1.0)
	(24, 1.0)
	(25, 1.0)
	(26, 1.0)
	(27, 1.0)
	(28, 1.0)
	(29, 1.0)
	(30, 1.0)
	(31, 1.0)
	(32, 1.0)
	(33, 1.0)
	(34, 1.0)
	(35, 1.0)
	(36, 1.0)
	(37, 1.0)
	(38, 1.0)
	(39, 1.0)
	(40, 1.0)		
};

\addplot[
color=teal,
mark=square,
mark size=1pt
]
coordinates {(1, 0.59304469)
	(2, 0.59601621)
	(3, 0.59710184)
	(4, 0.59752246)
	(5, 0.59831283)
	(6, 0.60190187)
	(7, 0.60576514)
	(8, 0.60606542)
	(9, 0.61057632)
	(10, 0.61317925)
	(11, 0.61885096)
	(12, 0.61970774)
	(13, 0.62018425)
	(14, 0.62113412)
	(15, 0.62129988)
	(16, 0.62305699)
	(17, 0.62747468)
	(18, 0.6309377)
	(19, 0.63734227)
	(20, 0.64378098)
	(21, 0.64840267)
	(22, 0.64963625)
	(23, 0.65100064)
	(24, 0.65367185)
	(25, 0.65884681)
	(26, 0.66404767)
	(27, 0.66543947)
	(28, 0.66714196)
	(29, 0.66717778)
	(30, 0.66877325)
	(31, 0.66912979)
	(32, 0.67065933)
	(33, 0.674473)
	(34, 0.67648715)
	(35, 0.67750777)
	(36, 0.67830708)
	(37, 0.67861267)
	(38, 0.67914229)
	(39, 0.68084083)
	(40, 0.6843787)		
};

\addplot[
color=orange,
mark=+,
mark size=2pt
]
coordinates {(1, 0.21737886)
	(2, 0.21778404)
	(3, 0.24830891)
	(4, 0.25283943)
	(5, 0.28877202)
	(6, 0.36720051)
	(7, 0.36790123)
	(8, 0.44587849)
	(9, 0.4728366)
	(10, 0.47582742)
	(11, 0.47793215)
	(12, 0.49422593)
	(13, 0.51518429)
	(14, 0.52260866)
	(15, 0.52649704)
	(16, 0.5456981)
	(17, 0.54675508)
	(18, 0.68836916)
	(19, 0.69255182)
	(20, 0.70124937)
	(21, 0.71681109)
	(22, 0.80578254)
	(23, 0.95729671)
	(24, 0.95911103)
	(25, 0.9608737)
	(26, 0.96235608)
	(27, 0.96285106)
	(28, 0.96331141)
	(29, 0.96378474)
	(30, 0.96496932)
	(31, 0.96610968)
	(32, 0.96874625)
	(33, 0.96902906)
	(34, 0.96929581)
	(35, 0.97415541)
	(36, 0.97476016)
	(37, 0.97551122)
	(38, 0.97716465)
	(39, 0.97778557)
	(40, 0.9840931)
		
};

\addplot[
color=black,
mark=triangle,
mark size=1pt
]
coordinates {(1, 0.20556897)
	(2, 0.31452946)
	(3, 0.58739952)
	(4, 0.63729011)
	(5, 0.64129289)
	(6, 0.7015062)
	(7, 0.8309141)
	(8, 0.8677253)
	(9, 0.88497346)
	(10, 0.88618872)
	(11, 0.92832667)
	(12, 0.94248955)
	(13, 0.94349213)
	(14, 0.9853238)
	(15, 0.98696513)
	(16, 0.990296)
	(17, 0.99047311)
	(18, 1.0)
	(19, 1.0)
	(20, 1.0)
	(21, 1.0)
	(22, 1.0)
	(23, 1.0)
	(24, 1.0)
	(25, 1.0)
	(26, 1.0)
	(27, 1.0)
	(28, 1.0)
	(29, 1.0)
	(30, 1.0)
	(31, 1.0)
	(32, 1.0)
	(33, 1.0)
	(34, 1.0)
	(35, 1.0)
	(36, 1.0)
	(37, 1.0)
	(38, 1.0)
	(39, 1.0)
	(40, 1.0)
};

\addplot[
color=violet,
mark=o,
mark size=1pt
]
coordinates {(1, 0.0)
	(2, 0.0)
	(3, 0.0)
	(4, 0.0)
	(5, 0.0)
	(6, 0.0)
	(7, 0.0)
	(8, 0.0)
	(9, 0.0)
	(10, 0.0)
	(11, 0.0)
	(12, 0.0)
	(13, 0.0)
	(14, 0.0)
	(15, 0.0)
	(16, 0.0)
	(17, 0.13377304)
	(18, 0.17276783)
	(19, 0.18112169)
	(20, 0.1878082)
	(21, 0.22232939)
	(22, 0.25722979)
	(23, 0.26168774)
	(24, 0.26907415)
	(25, 0.3041177)
	(26, 0.75702531)
	(27, 0.78852266)
	(28, 0.79619886)
	(29, 0.8098352)
	(30, 1.0)
	(31, 1.0)
	(32, 1.0)
	(33, 1.0)
	(34, 1.0)
	(35, 1.0)
	(36, 1.0)
	(37, 1.0)
	(38, 1.0)
	(39, 1.0)
	(40, 1.0)	
};

\addplot[
color=blue,
mark=o,
mark size=1.5pt
]
coordinates {(1, 0.55762042)
	(2, 0.55963476)
	(3, 0.56897393)
	(4, 0.5938201)
	(5, 0.59811671)
	(6, 0.60723987)
	(7, 0.60823221)
	(8, 0.62025834)
	(9, 0.62237208)
	(10, 0.62359097)
	(11, 0.63897223)
	(12, 0.64535994)
	(13, 0.65404255)
	(14, 0.65894628)
	(15, 0.66011205)
	(16, 0.6615146)
	(17, 0.66426739)
	(18, 0.66904969)
	(19, 0.67124623)
	(20, 0.68053737)
	(21, 0.68249606)
	(22, 0.68366847)
	(23, 0.68609597)
	(24, 0.68669861)
	(25, 0.68815054)
	(26, 0.69072155)
	(27, 0.6943825)
	(28, 0.699652)
	(29, 0.69991802)
	(30, 0.70090223)
	(31, 0.70770359)
	(32, 0.71359357)
	(33, 0.71647243)
	(34, 0.71715733)
	(35, 0.71835076)
	(36, 0.71836534)
	(37, 0.72536878)
	(38, 0.74624533)
	(39, 0.75024499)
	(40, 0.75084654)	
};
\end{axis}
\end{tikzpicture}\label{fig:fatrec_part_distribution}}
		\hfil
		\subfloat[{\bf Individual Spea. Satisfactions}]{\pgfplotsset{width=0.24\textwidth,height=0.23\textwidth,compat=1.9}
			\begin{tikzpicture}
\begin{axis}[
xlabel={Speakers},
ylabel={$NEC$},
xmin= 0.9, xmax=11.1,
ymin=-0.1, ymax=1.1,
xtick={1,5,10},
ytick={0,0.2,0.4,0.6,0.8,1},
ymajorgrids=true,
grid style=dashed
]
\addplot[
color=red,
mark=o,
mark size=1pt
]
coordinates {(1, 0.69938518)
	(2, 0.76650751)
	(3, 0.77689908)
	(4, 0.82113472)
	(5, 0.86664038)
	(6, 0.94713681)
	(7, 0.95674375)
	(8, 0.97813482)
	(9, 1)
	(10, 1)
	(11, 1)	
};

\addplot[
color=teal,
mark=square,
mark size=1pt
]
coordinates {	(1, 0.30377066)
	(2, 0.38238655)
	(3, 0.39637975)
	(4, 0.48184198)
	(5, 0.65564161)
	(6, 0.70630843)
	(7, 0.75880916)
	(8, 0.76228732)
	(9, 0.86664038)
	(10, 0.96582339)
	(11, 1.0)
};

\addplot[
color=orange,
mark=+,
mark size=2pt
]
coordinates {(1, 0.60926769)
	(2, 0.6720924)
	(3, 0.67403898)
	(4, 0.68659665)
	(5, 0.68788997)
	(6, 0.70449251)
	(7, 0.71010604)
	(8, 0.78194899)
	(9, 0.78634479)
	(10, 0.81320847)
	(11, 0.8147036)
};

\addplot[
color=black,
mark=triangle,
mark size=1pt
]
coordinates {(1, 0.69938518)
	(2, 0.7568297)
	(3, 0.77689908)
	(4, 0.81320847)
	(5, 0.86664038)
	(6, 0.92893966)
	(7, 0.94713681)
	(8, 0.96820614)
	(9, 0.99219826)
	(10, 0.99553773)
	(11, 1.0)
};

\addplot[
color=violet,
mark=o,
mark size=1pt
]
coordinates {(1, 0.34341429)
	(2, 0.35504806)
	(3, 0.3642569)
	(4, 0.37252576)
	(5, 0.38238655)
	(6, 0.39637975)
	(7, 0.39859711)
	(8, 0.41148822)
	(9, 0.42625964)
	(10, 0.45770886)
	(11, 0.58986725)
};

\addplot[
color=blue,
mark=o,
mark size=1.5pt
]
coordinates {(1, 0.56274688)
	(2, 0.59311497)
	(3, 0.66307051)
	(4, 0.6720924)
	(5, 0.68659665)
	(6, 0.70449251)
	(7, 0.7089098)
	(8, 0.73585477)
	(9, 0.76650751)
	(10, 0.82650965)
	(11, 0.84866596)
};
\end{axis}
\end{tikzpicture}\label{fig:fatrec_spea_distribution}}
		\vfil
		\vspace{-3mm}
		\subfloat{\pgfplotsset{width=.7\textwidth,compat=1.9}
			\begin{tikzpicture}
			\begin{customlegend}[legend entries={{\bf PFair},{\bf EM},{\bf SFair},{\bf IAM}, {\bf OfflineConf}, {\bf mFairConf ($\lambda_1=\lambda_2=0.5$)}},legend columns=2,legend style={/tikz/every even column/.append style={column sep=0.5cm}}]
			\addlegendimage{teal,mark=square,mark size=1pt,sharp plot}
			\addlegendimage{red,mark=o,mark size=1pt,sharp plot}
			\addlegendimage{orange,mark=+,mark size=2pt,sharp plot}
			\addlegendimage{black,mark=triangle,mark size=1pt,sharp plot}
			\addlegendimage{violet,mark=o,mark size=1pt,sharp plot}
			\addlegendimage{blue,mark=o,mark size=1.5pt,sharp plot}
			
			\end{customlegend}
			\end{tikzpicture}}
	}
	\caption{Individual participant and speaker satisfactions (sorted in increasing order) in FATREC dataset.}\label{fig:fatrec_distribution}
\end{figure}
\begin{table}\footnotesize
	\begin{tabular}{|c|c|c|c|c|}
		\hline
		{\bf Methods} & \vtop{\hbox{\strut $NCG_\text{max}$}\hbox{\strut $-NCG_\text{min}$}} & $NCG_\text{mean}$& \vtop{\hbox{\strut $NEC_\text{max}$}\hbox{\strut $-NEC_\text{min}$}} & $NEC_\text{mean}$\\\cline{1-5}
		PFair&$0.09$&$0.64$&$0.69$&$0.66$\\\cline{1-5}
		SFair&$0.76$&$0.70$&$0.20$&$0.72$\\\cline{1-5}
		EM&$0.78$&$0.91$&$0.30$&$0.89$\\\cline{1-5}
		IAM&$0.79$&$0.90$&$0.30$&$0.88$\\\cline{1-5}
		OfflineConf&$1$&$0.40$&$0.24$&$0.40$\\\cline{1-5}
		\vtop{\hbox{\strut mFairConf}\hbox{\strut ($\lambda_1=\lambda_2=0.5$)}} &$0.19$&$0.66$&$0.28$&$0.70$\\\cline{1-5}
	\end{tabular}
	\caption{\bf FATREC results (rounded upto two decimal points).}
	\label{tab:fatrec_results}
\end{table}
\begin{figure}
	\center{\small
		\subfloat[{\bf Individual Part. Satisfactions}]{\pgfplotsset{width=0.24\textwidth,height=0.23\textwidth,compat=1.9}
			\input{figures/recsys_part_distribution}\label{fig:recsys_part_distribution}}
		\hfil
		\subfloat[{\bf Individual Spea. Satisfactions}]{\pgfplotsset{width=0.24\textwidth,height=0.23\textwidth,compat=1.9}
			\begin{tikzpicture}
\begin{axis}[
xlabel={Speakers},
ylabel={$NEC$},
xmin=0.9, xmax=26.1,
ymin=0.2, ymax=1.01,
xtick={1,5,10,15,20,25},
ytick={0,0.2,0.4,0.6,0.8,1},
ymajorgrids=true,
grid style=dashed
]

\addplot[
color=red,
mark=*,mark size=0.3pt
]
coordinates {(1, 0.8008849557522124)
	(2, 0.8235294117647058)
	(3, 0.8333333333333334)
	(4, 0.8356807511737089)
	(5, 0.8469387755102041)
	(6, 0.8513513513513513)
	(7, 0.8614718614718615)
	(8, 0.8666666666666667)
	(9, 0.8796296296296297)
	(10, 0.8867924528301887)
	(11, 0.889261744966443)
	(12, 0.889795918367347)
	(13, 0.8952380952380953)
	(14, 0.9045643153526971)
	(15, 0.9074074074074074)
	(16, 0.9077669902912622)
	(17, 0.9116279069767442)
	(18, 0.9224137931034483)
	(19, 0.9256900212314225)
	(20, 0.9422222222222222)
	(21, 0.95260663507109)
	(22, 0.9596774193548387)
	(23, 0.97)
	(24, 0.9734939759036144)
	(25, 1.0)
	(26, 1.0)		
};

\addplot[
color=teal,
mark=*, mark size=0.3pt
]
coordinates {(1, 0.23773584905660378)
	(2, 0.25675675675675674)
	(3, 0.27555555555555555)
	(4, 0.3148148148148148)
	(5, 0.3472222222222222)
	(6, 0.39805825242718446)
	(7, 0.40476190476190477)
	(8, 0.46938775510204084)
	(9, 0.504424778761062)
	(10, 0.535)
	(11, 0.6122448979591837)
	(12, 0.6559139784946236)
	(13, 0.6982758620689655)
	(14, 0.802013422818792)
	(15, 0.8093023255813954)
	(16, 0.8095238095238095)
	(17, 0.8144796380090498)
	(18, 0.8354978354978355)
	(19, 0.8356807511737089)
	(20, 0.8459302325581395)
	(21, 0.8776978417266187)
	(22, 0.8815165876777251)
	(23, 0.8895966029723992)
	(24, 0.9144144144144144)
	(25, 0.9734939759036144)
	(26, 1.0)
		
};

\addplot[
color=orange,
mark=*, mark size=0.3pt
]
coordinates {(1, 0.8815165876777251)
	(2, 0.8828828828828829)
	(3, 0.8828828828828829)
	(4, 0.8844444444444445)
	(5, 0.8867924528301887)
	(6, 0.8895966029723992)
	(7, 0.889795918367347)
	(8, 0.8917748917748918)
	(9, 0.892018779342723)
	(10, 0.8926174496644296)
	(11, 0.8930232558139535)
	(12, 0.895)
	(13, 0.8952380952380953)
	(14, 0.8959276018099548)
	(15, 0.8962962962962963)
	(16, 0.8974820143884892)
	(17, 0.8982300884955752)
	(18, 0.8987951807228916)
	(19, 0.9)
	(20, 0.9008620689655172)
	(21, 0.9045643153526971)
	(22, 0.9059139784946236)
	(23, 0.9069767441860465)
	(24, 0.9074074074074074)
	(25, 0.9077669902912622)
	(26, 0.9081632653061225)		
};

\addplot[
color=black,
mark=*, mark size=0.3pt
]
coordinates {(1, 0.7572815533980582)
	(2, 0.7647058823529411)
	(3, 0.7904761904761904)
	(4, 0.7972972972972973)
	(5, 0.8095238095238095)
	(6, 0.8177777777777778)
	(7, 0.8232558139534883)
	(8, 0.8311688311688312)
	(9, 0.8403755868544601)
	(10, 0.8423236514522822)
	(11, 0.8459302325581395)
	(12, 0.8468468468468469)
	(13, 0.8469387755102041)
	(14, 0.8530612244897959)
	(15, 0.855)
	(16, 0.8663793103448276)
	(17, 0.8716814159292036)
	(18, 0.8716981132075472)
	(19, 0.8796296296296297)
	(20, 0.9074074074074074)
	(21, 0.9241706161137441)
	(22, 0.9261744966442953)
	(23, 0.9596774193548387)
	(24, 0.9734939759036144)
	(25, 1.0)
	(26, 1.0)
};

\addplot[
color=blue,
mark=*, mark size=0.3pt
]
coordinates {(1, 0.49333333333333335)
	(2, 0.505)
	(3, 0.7123655913978495)
	(4, 0.7965367965367965)
	(5, 0.8055555555555556)
	(6, 0.8093023255813954)
	(7, 0.819047619047619)
	(8, 0.8423423423423423)
	(9, 0.8554216867469879)
	(10, 0.8867924528301887)
	(11, 0.8874734607218684)
	(12, 0.8952380952380953)
	(13, 0.8993288590604027)
	(14, 0.9045643153526971)
	(15, 0.9074074074074074)
	(16, 0.9142857142857143)
	(17, 0.9154929577464789)
	(18, 0.9224137931034483)
	(19, 0.9279279279279279)
	(20, 0.9336283185840708)
	(21, 0.9370503597122302)
	(22, 0.9387755102040817)
	(23, 0.95260663507109)
	(24, 0.9563106796116505)
	(25, 1.0)
	(26, 1.0)
};
\end{axis}
\end{tikzpicture}\label{fig:recsys_spea_distribution}}
		\vfil
		\vspace{-3mm}
		\subfloat{\pgfplotsset{width=.7\textwidth,compat=1.9}
			\begin{tikzpicture}
			\begin{customlegend}[legend entries={{\bf PFair-RRFS},{\bf EM},{\bf SFair-RRFS},{\bf IAM}, {\bf mFairConf-RRFS ($\lambda_1=\lambda_2=0.05$)}},legend columns=2,legend style={/tikz/every even column/.append style={column sep=0.5cm}}]
			\addlegendimage{teal,mark=*,mark size=0.3pt,sharp plot}
			\addlegendimage{red,mark=*,mark size=0.3pt,sharp plot}
			\addlegendimage{orange,mark=*,mark size=0.3pt,sharp plot}
			\addlegendimage{black,mark=*,mark size=0.3pt,sharp plot}
			\addlegendimage{blue,mark=*,mark size=0.3pt,sharp plot}
			
			\end{customlegend}
			\end{tikzpicture}}
	}
	\caption{Individual participant and speaker satisfactions (sorted in increasing order) in RECSYS dataset.}\label{fig:recsys_distribution}
\end{figure}
\begin{table}\footnotesize
	\begin{tabular}{|c|c|c|c|c|}
		\hline
		{\bf Methods} & $NCG_\text{gini}$ & $NCG_\text{mean}$& $NEC_\text{gini}$ & $NEC_\text{mean}$\\\cline{1-5}
		PFair-RRFS&$0.17$&$0.34$&$0.21$&$0.65$\\\cline{1-5}
		SFair-RRFS&$0.26$&$0.45$&$0.01$&$0.89$\\\cline{1-5}
		EM&$0.23$&$0.46$&$0.03$&$0.9$\\\cline{1-5}
		IAM&$0.25$&$0.44$&$0.04$&$0.86$\\\cline{1-5}
		\vtop{\hbox{\strut mFairConf-RRFS}\hbox{\strut ($\lambda_1=0.05,\lambda_2=0$)}} &$0.18$&$0.43$&$0.07$&$0.85$\\\cline{1-5}
		\vtop{\hbox{\strut mFairConf-RRFS}\hbox{\strut ($\lambda_1=\lambda_2=0.05$)}} &$0.19$&$0.43$&$0.06$&$0.87$\\\cline{1-5}
		\vtop{\hbox{\strut mFairConf-RRFS}\hbox{\strut ($\lambda_1=0,\lambda_2=0.05$)}} &$0.22$&$0.45$&$0.03$&$0.9$ \\\cline{1-5}
	\end{tabular}
	\caption{RECSYS results (rounded upto two decimal places).}
	\label{tab:recsys_results}
\end{table}
\begin{figure}[t!]
	\center{\small
		\subfloat[{\bf FATREC}]{\pgfplotsset{width=0.24\textwidth,height=0.17\textwidth,compat=1.9}
			\begin{tikzpicture}
\begin{axis}[
xlabel={\#ContiguousTalks},
ylabel={Frequency},
xmin=0.5, xmax=10.5,
ymin=0, ymax=6,
xtick={1,3,5,7,9},
ytick={1,2,3,4,5},
xticklabel style={rotate=60},
ymajorgrids=true,
grid style=dashed
]

\addplot[
color=red,
mark=o,
mark size=2pt
]
coordinates {(1,1)
	(3,1)
	(7,1)
};

\addplot[
color=teal,
mark=square,
mark size=2pt
]
coordinates {(1,3)
	(2,4)		
};

\addplot[
color=orange,
mark=+,
mark size=3pt
]
coordinates {(1,1)
	(2,2)
	(8,1)
};

\addplot[
color=black,
mark=triangle,
mark size=2pt
]
coordinates {(1,1)
	(10,1)
};

\addplot[
color=blue,
mark=x,
mark size=3pt
]
coordinates {(3,1)
	(4,2)
};
\end{axis}
\end{tikzpicture}\label{fig:fatrec_contiguity}}
		\hfil
		\subfloat[{\bf RECSYS}]{\pgfplotsset{width=0.24\textwidth,height=0.17\textwidth,compat=1.9}
			\begin{tikzpicture}
\begin{axis}[
xlabel={\#ContiguousTalks},
ylabel={Frequency},
xmin=0.5, xmax=17.5,
ymin=0, ymax=10,
xtick={1,3,5,7,9,11,13,15,17},
ytick={2,4,6,8,10},
xticklabel style={rotate=60},
ymajorgrids=true,
grid style=dashed
]

\addplot[
color=red,
mark=o,
mark size=2pt
]
coordinates {(10,1)
	(16,1)
};

\addplot[
color=teal,
mark=square,
mark size=2pt
]
coordinates {(1,8)
	(2,3)
	(3,1)
	(4,1)
	(5,1)	
};

\addplot[
color=orange,
mark=+,
mark size=3pt
]
coordinates {(1,3)
	(2,4)
	(3,1)
	(5,1)
	(7,1)
};

\addplot[
color=black,
mark=triangle,
mark size=2pt
]
coordinates {(10,1)
	(16,1)
};

\addplot[
color=blue,
mark=x,
mark size=3pt
]
coordinates {(3,1)
	(5,1)
	(8,1)
	(10,1)
};
\end{axis}
\end{tikzpicture}\label{fig:recsys_contiguity}}
		\vfil
		\vspace{-3mm}
		\subfloat{\pgfplotsset{width=.7\textwidth,compat=1.9}
			\begin{tikzpicture}
			\begin{customlegend}[legend entries={{\bf EM},{\bf PFair},{\bf SFair},{\bf IAM}, {\bf mFairConf}},legend columns=5,legend style={/tikz/every even column/.append style={column sep=0.2cm}}]
			\addlegendimage{red,mark=o,mark size=2pt,sharp plot}
			\addlegendimage{teal,mark=square,mark size=2pt,sharp plot}
			\addlegendimage{orange,mark=+,mark size=3pt,sharp plot}
			\addlegendimage{black,mark=triangle,mark size=2pt,sharp plot}
			\addlegendimage{blue,mark=x,mark size=3pt,sharp plot}
			
			\end{customlegend}
			\end{tikzpicture}}
	}
	\caption{Schedule Contiguity. (FATREC: mFairConf with ($\lambda_1=\lambda_2=0.5$); RECSYS: mFairConf uses RRFS with ($\lambda_1=\lambda_2=0.05$), PFair and SFair use RRFS)}\label{fig:sch_contiguity}
\end{figure}
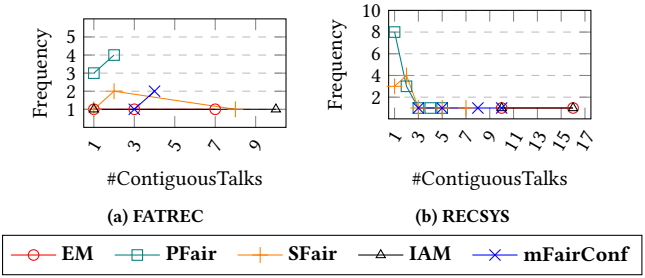
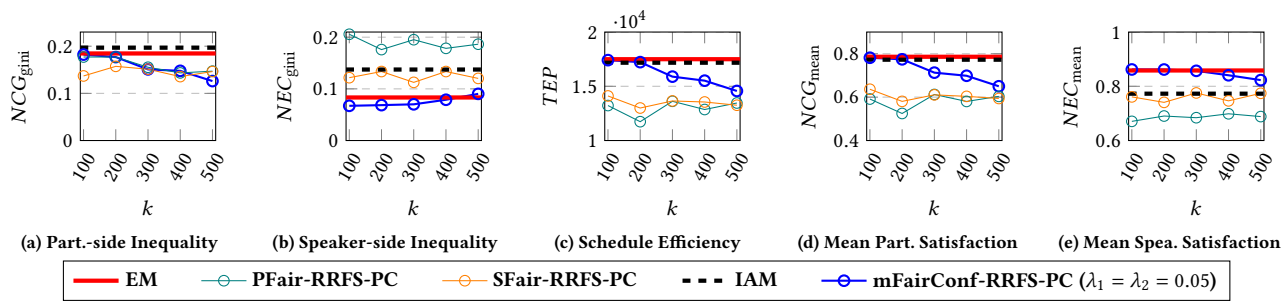
\begin{figure*}
	\center{\small
		\subfloat[{\bf Part.-side Inequality}]{\pgfplotsset{width=0.19\textwidth,height=0.17\textwidth,compat=1.9}
			\begin{tikzpicture}
\begin{axis}[
xlabel={$k$},
ylabel={$NCG_\text{gini}$},
xmin=90, xmax=510,
ymin=0, ymax=0.23,
xtick={100,200,300,400,500},
ytick={0,0.1,0.2},
xticklabel style={rotate=60},
ymajorgrids=true,
grid style=dashed,
]

\addplot[
color=blue,
mark=o,
mark size=2pt,
thick
]
coordinates {(100, 0.18220807996897823)
	(200, 0.17693505687095853)
	(300, 0.15080720048404309)
	(400, 0.14773595666465109)
	(500, 0.12605439526313428)
};

\addplot[
color=teal,
mark=o,
mark size=2pt
]
coordinates {(100,0.17706932673197212)
	(200,0.17696443576135876)
	(300,0.15535809804217765)
	(400,0.14334070063587484)
	(500,0.14705609509891818)
};

\addplot[
color=orange,
mark=o,
mark size=2pt
]
coordinates {(100,0.13689447225776197)
	(200,0.15726301142357108)
	(300,0.15123524428823765)
	(400,0.1355163183693896)
	(500,0.14623131462113123)
};

\addplot[color=red,mark=none,ultra thick] coordinates {(90,0.18457974566469387) (510,0.18457974566469387)};

\addplot[color=black,mark=none,dashed,ultra thick] coordinates {(90,0.19682159919542266) (510,0.19682159919542266)};

\end{axis}
\end{tikzpicture}\label{fig:icml_part_gini}}
		\hfil
		\subfloat[{\bf Speaker-side Inequality}]{\pgfplotsset{width=0.19\textwidth,height=0.17\textwidth,compat=1.9}
			\begin{tikzpicture}
\begin{axis}[
xlabel={$k$},
ylabel={$NEC_\text{gini}$},
xmin=90, xmax=510,
ymin=0, ymax=0.21,
xtick={100,200,300,400,500},
ytick={0,0.10,0.2},
xticklabel style={rotate=60},
ymajorgrids=true,
grid style=dashed,
]

\addplot[
color=blue,
mark=o,
mark size=2pt,
thick
]
coordinates {(100, 0.0671864147548843)
	(200, 0.06852621799723439)
	(300, 0.07008865593158717)
	(400, 0.07891462340316746)
	(500, 0.09032979466743614)
};

\addplot[
color=teal,
mark=o,
mark size=2pt
]
coordinates {
	(100,0.20570698444979488)
	(200,0.17602618636731537)
	(300,0.19516771366102217)
	(400,0.1784189007950624)
	(500,0.18644451558922345)
};

\addplot[
color=orange,
mark=o,
mark size=2pt
]
coordinates {(100,0.12163191526199249)
	(200,0.13354785731549818)
	(300,0.11258096830449146)
	(400,0.13354808660550385)
	(500,0.12024979725239987)
};

\addplot[color=red,mark=none,ultra thick] coordinates {(90,0.0833253403466467) (510,0.0833253403466467)};

\addplot[color=black,mark=none,dashed,ultra thick] coordinates {(90,0.13742522173668992) (510,0.13742522173668992)};

\end{axis}
\end{tikzpicture}\label{fig:icml_spea_gini}}
		\hfil
		\subfloat[{\bf Schedule Efficiency}]{\pgfplotsset{width=0.19\textwidth,height=0.17\textwidth,compat=1.9}
			\begin{tikzpicture}
\begin{axis}[
xlabel={$k$},
ylabel={$TEP$},
xmin=90, xmax=510,
ymin=10000, ymax=20000,
xtick={100,200,300,400,500},
ytick={10000,15000,20000},
xticklabel style={rotate=60},
ymajorgrids=true,
grid style=dashed,
]

\addplot[
color=blue,
mark=o,
mark size=2pt,
thick
]
coordinates {(100, 17409.5)
	(200, 17225.5)
	(300, 15889.5)
	(400, 15518.0)
	(500, 14568.0)
};

\addplot[
color=teal,
mark=o,
mark size=2pt
]
coordinates {(100,13225.5)
	(200,11736.5)
	(300,13616.0)
	(400,12859.0)
	(500,13459.5)
};

\addplot[
color=orange,
mark=o,
mark size=2pt
]
coordinates {(100,14109.0)
	(200,12996.5)
	(300,13647.5)
	(400,13552.0)
	(500,13254.0)
};

\addplot[color=red,mark=none,ultra thick] coordinates {(90,17496.5) (510,17496.5)};

\addplot[color=black,mark=none,dashed,ultra thick] coordinates {(90,17172.0) (510,17172.0)};

\end{axis}
\end{tikzpicture}\label{fig:icml_TEP}}
		\hfil
		\subfloat[{\bf Mean Part. Satisfaction}]{\pgfplotsset{width=0.19\textwidth,height=0.17\textwidth,compat=1.9}
			\begin{tikzpicture}
\begin{axis}[
xlabel={$k$},
ylabel={$NCG_\text{mean}$},
xmin=90, xmax=510,
ymin=0.4, ymax=0.9,
xtick={100,200,300,400,500},
ytick={0,0.2,0.4,0.6,0.8,1},
xticklabel style={rotate=60},
ymajorgrids=true,
grid style=dashed,
]

\addplot[
color=blue,
mark=o,
mark size=2pt,
thick
]
coordinates {(100, 0.7815519110925847)
	(200, 0.7742402872957673)
	(300, 0.7128867895440518)
	(400, 0.6984747658381525)
	(500, 0.6500015123322617)
};

\addplot[
color=teal,
mark=o,
mark size=2pt
]
coordinates {
	(100,0.5912157730911669)
	(200,0.5246516920700858)
	(300,0.6122691381028876)
	(400,0.5801470927028247)
	(500,0.6034497445469817)
};

\addplot[
color=orange,
mark=o,
mark size=2pt
]
coordinates {(100,0.6363197743069184)
	(200,0.5804207132092036)
	(300,0.6105212709725404)
	(400,0.6035279680178174)
	(500,0.5937377688760399)
};

\addplot[color=red,mark=none,ultra thick] coordinates {(90,0.7859990733764628) (510,0.7859990733764628)};

\addplot[color=black,mark=none,dashed,ultra thick] coordinates {(90,0.7724900679578821) (510,0.7724900679578821)};

\end{axis}
\end{tikzpicture}\label{fig:icml_part_mu}}
		\hfil
		\subfloat[{\bf Mean Spea. Satisfaction}]{\pgfplotsset{width=0.19\textwidth,height=0.17\textwidth,compat=1.9}
			\begin{tikzpicture}
\begin{axis}[
xlabel={$k$},
ylabel={$NEC_\text{mean}$},
xmin=90, xmax=510,
ymin=0.6, ymax=1,
xtick={100,200,300,400,500},
ytick={0,0.2,0.4,0.6,0.8,1},
xticklabel style={rotate=60},
ymajorgrids=true,
grid style=dashed,
]

\addplot[
color=blue,
mark=o,
mark size=2pt,
thick
]
coordinates {(100, 0.8619702052659796)
	(200, 0.8618111221871017)
	(300, 0.8569253545861188)
	(400, 0.8406335949927701)
	(500, 0.8227652215058437)
};

\addplot[
color=teal,
mark=o,
mark size=2pt
]
coordinates {(100,0.6706600499628991)
	(200,0.6903311515322252)
	(300,0.684149492723613)
	(400,0.6986053741872524)
	(500,0.6882813653186045)
};

\addplot[
color=orange,
mark=o,
mark size=2pt
]
coordinates {(100,0.7609316072316639)
	(200,0.7406867017611013)
	(300,0.7754301414580076)
	(400,0.7462496568030699)
	(500,0.7741019090484557)
};

\addplot[color=red,ultra thick]{0.8582527688208259};
\addplot[color=red,mark=none,ultra thick] coordinates {(90,0.8582527688208259) (510,0.8582527688208259)};

\addplot[color=black,mark=none,dashed,ultra thick] coordinates {(90,0.7726861813537671) (510,0.7726861813537671)};

\end{axis}
\end{tikzpicture}\label{fig:icml_spea_mu}}		
		\vfil
		\vspace{-3mm}
		\subfloat{\pgfplotsset{width=.7\textwidth,compat=1.9}
			\begin{tikzpicture}
			\begin{customlegend}[legend entries={{\bf EM},{\bf PFair-RRFS-PC},{\bf SFair-RRFS-PC},{\bf IAM},{\bf mFairConf-RRFS-PC ($\lambda_1=\lambda_2=0.05$)}},legend columns=5,legend style={/tikz/every even column/.append style={column sep=0.5cm}}]
			\addlegendimage{red,mark=.,ultra thick,sharp plot}
			\addlegendimage{teal,mark=o,mark size=2pt,sharp plot}
			\addlegendimage{orange,mark=o,mark size=2pt,sharp plot}
			\addlegendimage{dashed,black,mark=.,ultra thick,sharp plot}
			\addlegendimage{blue,mark=o,thick,sharp plot}
			\end{customlegend}
			\end{tikzpicture}}
	}\vspace{-2mm}
	\caption{Results on ICML dataset.}\label{fig:icml_cluster_variations}
	\vspace{-2mm}
\end{figure*}
\subsubsection{\bf Results on FATREC Dataset:}\label{subsubsec:FATREC}
In fig-\ref{fig:fatrec_part_distribution} and fig-\ref{fig:fatrec_spea_distribution}, we plot the individual participant satisfactions and speaker satisfactions---both sorted in increasing order---in FATREC dataset.
We list relevant metric values in tab-\ref{tab:fatrec_results}.
Along with the baselines and mFairConf, we also evaluate the real FATREC-2017 workshop schedule (referred to as OfflineConf).

\noindent {\bf Baseline Results:}
While EM and IAM are able to ensure fairness on speaker-side (low inequality in $NEC$s: tab-\ref{tab:fatrec_results}), they cause huge participant unfairness ($NCG_\text{max}-NCG_\text{max}$ in tab-\ref{tab:fatrec_results}).
This is because, the number of talks ($11$) is very few in comparison to a total of $96$ available slots, and there are enough number of slots favorable to crowds from either European or American timezones---covering majority of participants;
so both EM and IAM are able to achieve high satisfaction for all speakers by just scheduling their talks in the slots favorable to either of the majority participant groups while undermining the minority participant group from non-European and non-American timezones.
Similarly, SFair also undermines the minority participant group (fig-\ref{fig:fatrec_part_distribution}).
On the other hand, PFair significantly flattens the participant-side curve (fig-\ref{fig:fatrec_part_distribution}) thereby being the most fair for participants;
however it comes at a price---huge unfairness on the speaker-side (fig-\ref{fig:fatrec_spea_distribution}, tab-\ref{tab:fatrec_results}).

\noindent {\bf OfflineConf and mFairConf Results:}
As FATREC-2017 was held in-person in Como, Italy, its schedule was local timezone-specific.
If the same schedule were to be used in case of an online version, it would favour the participants mostly from nearby European timezones while severely undermining the participants from distant timezones (refer OfflineConf in fig-\ref{fig:fatrec_part_distribution}).
Such timezone-specific schedule also leaves no chance for the talks to be scheduled in slots with optimal availability of other majority participant groups, which then leads to highly suboptimal and unfair schedule for the speakers too (refer OfflineConf in tab-\ref{tab:fatrec_results}).
While the baselines and timezone-specific schedule are proving to be less suitable for online conferences, mFairConf, on the other hand, with $\lambda_1=\lambda_2=0.5$ setting strikes a good balance by significantly reducing the max-min gap for the participants---thereby improving inclusivity--- while still maintaining good speaker satisfaction and fairness similar to EM and SFair (tab-\ref{tab:fatrec_results}).
\subsubsection{\bf Results on RECSYS Dataset:}\label{subsubsec:RECSYS}
While the number of talks in RECSYS is small, it has a very high number of participants.
Thus, we use the proposed scalable approach (RRFS as in sec-\ref{subsubsec:RRFS}).
Note that, due to the hardness of PFair and SFair, we compute them also through repeated rounding of fractional solutions.
Similar to FATREC, here, we plot the individual participant and speaker satisfactions in fig-\ref{fig:recsys_distribution}, and metric values in tab-\ref{tab:recsys_results}.
Both EM and IAM result in schedules with high inequality for the participants (fig-\ref{fig:recsys_part_distribution}), i.e., high participant unfairness (tab-\ref{tab:recsys_results}).
While PFair reduces the inequality on the participant-side, it increases the speaker-side inequality (fig-\ref{fig:recsys_spea_distribution});
SFair behaves completely the opposite way.
In this dataset, we find  a balanced performance from mFairConf with $\lambda_1=\lambda_2=0.05$, i.e., reduction in max-min gaps on participants and speaker sides (also low gini index) without much degradation in overall satisfactions (tab-\ref{tab:recsys_results}).
This also serves as empirical evidence for the efficacies of scalable RRFS approach for big conferences.
\subsubsection{\bf Schedule Contiguity Analysis:}
In physical conferences, the talks are  clubbed together in smaller numbers and scheduled in  a contiguous manner.
Such a strict contiguity need not be maintained in the virtual setting; however, 
if the talks are scheduled in a very segregated manner, participants might lose interest.
Thus, we plot the frequencies of differently clubbed contiguous talks in fig-\ref{fig:sch_contiguity}.
We find that PFair results in a very segregated (high number of singular talks) schedule as it tries to fairly satisfy participants from all timezones.
However, both EM and IAM give schedules where most of the talks are clustered around time slots favorable to majority of participants.
SFair, however, leads to schedules with a mix of large number of contiguous talks and a few number of singular talks.
mFairConf, on the other hand, cares about efficiency and fairness simultaneously, thus clubs small number of talks and schedules them in the slots which are often at the intersection of availability intervals of participants from different timezones.
\subsubsection{\bf Results on ICML Dataset:}\label{subsubsec:ICML}
As ICML has a very large number of participants, we first cluster the participants as proposed in sec-\ref{subsubsec:part_clustering}, and then apply RRFS (as in sec-\ref{subsubsec:RRFS}) for mFairConf joint optimization.
In fig-\ref{fig:icml_cluster_variations}, we plot the results on ICML dataset.
Similar to previous datasets, here also, we observe mFairConf to be performing in a more balanced manner than the baselines.
We capture the changes in performance by varying $k$ (number of participant clusters).
While we observe no specific trend in the performances of the baselines by varying $k$, there is a clear trend in mFairConf's performance;
with other settings fixed, an increase in $k$ results in an increase participant fairness (decreasing inequality in fig-\ref{fig:icml_part_gini}), and subsequently a decrease in speaker fairness (increasing inequality in fig-\ref{fig:icml_spea_gini}).
This is because higher $k$ leads to smaller participant clusters where the centroids are able to better represent all the participants leading to better participant fairness which shifts some talks to the favorable slots of previously less-represented participants thereby decreasing speaker fairness.
Note that an increase in participant fairness, here, also comes with losses in efficiency, and individual participant and speaker satisfactions (figs-\ref{fig:icml_TEP},\ref{fig:icml_part_mu},\ref{fig:icml_spea_mu}).
\subsubsection{\bf Priority Scheduling and Repetitions:}\label{subsubsec:priority_experiments}
Sometimes the conference organizers might have varying priorities towards the talks (e.g. short vs. long, main vs. special track).
In such cases, multiple instances of mFairConf in asynchronous setup can be used to schedule different groups of talks based on their priority levels.
Top-priority talks can be scheduled first by allowing their mapping to any possible slot, and followed by scheduling lower priority talks in the remaining slots.
This allows us to bring intra-(priority)-group fairness for speakers.
Moreover such an asynchronous setup can also allow for repetitions of top-priority talks once a round of scheduling is done.
We test such priority scheduling and repetitions on RECSYS dataset.
We group the talks into three priority levels (top, medium and low) of equal size based on their overall interest scores, and then asynchronously schedule them.
While detailed results are in appendix tab-\ref{tab:priority_scheduling}, we highlight some important findings.
(i) In comparison to a full-scale mFairConf, priority scheduling achieves better participant satisfaction since the top level talks get slots with more participant availability as they do not have to compete with low priority talks any more, and it also gives better intra-group fairness for speakers.
(ii) With repetitions of talks, participant satisfaction slowly increases and the schedules get more fair for participants.
(iii) Priority-based repetitions also increase the total expected crowd at the talks thereby increasing speaker satisfaction.
In fact, with prioritized repetitions of talks, the speakers may even have access to more audience ($NEC$ more than $1$) than what they would have received in their single best slot.
Even though the gaps between two assigned slots for a talk can be upto $12$ hours, the speakers often have incentive of getting more total audience through repetitions.
Note that with repetitions, the inequality on speaker-side might increase since the set of good slots which can ensure speaker fairness would already have been allocated in the first schedule, and the remaining slots, to be allocated for repetitions, may have varying effects on speaker satisfactions.
\section{Conclusion}\label{sec:discussion}
In this work, we modeled a very timely and important problem of virtual conference scheduling with efficiency and fairness concerns.
Apart from the formal definitions, we brought out fundamental tensions among participant fairness, speaker fairness, and efficiency. 
%
We experimentally showed that the proposed joint optimization framework, mFairConf, can find balanced conference schedules and generate schedules as per an organizer's relative priorities towards participants and speakers.
We note some of the limitations of the present work and possible future directions in the appendix.
~\\{\bf Project Repository:}
\url{https://github.com/gourabkumarpatro/FairConf}.
~\\{\bf Acknowledgments:} G. K Patro acknowledges the support by TCS Research Fellowship. This research was supported in part by ERC Grants for ``Foundations for Fair Social Computing" (agreement no. 789373), and ``NoBIAS - Artificial Intelligence without Bias" (agreement no. 860630) funded under the EU's Horizon 2020.

\balance
\bibliographystyle{ACM-Reference-Format}
\bibliography{Main}
\clearpage
\appendix
\begin{table*}[t]\footnotesize
	\caption{Examples}
	\vspace{-3mm}
	\subfloat[\bf Example Problem 1]{\begin{tabular}{*{5}{|c|c|c|c|c}p{5em}}
		\hline
		\multirow{2}*{\bf Participants} & \multicolumn{1}{c|}{\bf $V_p(t)$} & \multicolumn{3}{c|}{\bf $A_p(s)$}\\\cline{2-5}
		& $t$ & $s_1$ & $s_2$& $s_3$\\\cline{1-5}
		$p_1$ & $1$ & $1$& $0.49$ &$0$ \\\cline{1-5}
		$p_2$  & $1$ &$0$ & $0.49$&$1$\\\cline{1-5}   
	\end{tabular}
	\label{tab:toy_example_1}}
	\hfil
	\subfloat[\bf Example Problem 2]{\begin{tabular}{*{6}{|c|c|c|c|c|c}p{5em}}
		\hline
		\multirow{2}*{\bf Participants} & \multicolumn{2}{c|}{\bf $V_p(t)$} & \multicolumn{3}{c|}{\bf $A_p(s)$}\\\cline{2-6}
		& $t_1$ & $t_2$ & $s_1$ & $s_2$& $s_3$\\\cline{1-6}
		$p$ & $1$&$0.5$ &$1$ &$0.75$ &$0.8$ \\\cline{1-6}
	\end{tabular}
	\label{tab:toy_example_2}}
	\hfil
	\subfloat[\bf Example Problem 3]{\begin{tabular}{*{7}{|c|c|c|c|c}p{5em}}
			\hline
			\multirow{2}*{\bf Participants} & \multicolumn{2}{c|}{\bf $V_p(t)$} & \multicolumn{4}{c|}{\bf $A_p(s)$}\\\cline{2-7}
			& $t_1$ & $t_2$ & $s_1$ & $s_2$& $s_3$ & $s_4$\\\cline{1-7}
			$p_1$ & $1$ & $0.7$ & $1$ & $1$ & $0$ & $0.2$ \\\cline{1-7}
			$p_2$ & $1$ & $0.7$ & $1$ & $0$ & $1$ & $0.2$ \\\cline{1-7}   
		\end{tabular}
		\label{tab:toy_example_3}}
	\vspace{-3mm}
\end{table*}
\section{Appendix}	
{\bf Proof of Lemma-\ref{lemma:swm_max_matching}:}
Let's use matrix $X$ of dimensions $\abs{\mathcal{T}}\times \abs{\mathcal{S}}$ to represent a conference schedule; 
i.e., the element $X_{t,s}\in\{0,1\}$ is a binary indicator variable for talk $t\in \mathcal{T}$ being scheduled in slot $s\in \mathcal{S}$.
We can now rewrite the efficiency maximization problem ({\small $\argmax_\Gamma \sum\limits_{t\in\mathcal{T}}\sum\limits_{p\in\mathcal{P}} V_p(t)\times A_p(\Gamma(t))$}) as below.
{\small\[\argmax_X\sum\limits_{t\in\mathcal{T}}\sum\limits_{p\in\mathcal{P}} \sum\limits_{s\in\mathcal{S}}V_p(t)\cdot A_p(s)\cdot X_{t,s}\]
\[\equiv \argmax_X \sum\limits_{t\in\mathcal{T}}\sum\limits_{s\in\mathcal{S}}\sum\limits_{p\in\mathcal{P}}V_p(t)\cdot A_p(s)\cdot X_{t,s}\]
\[\equiv \argmax_X \sum\limits_{t\in\mathcal{T}}\sum\limits_{s\in\mathcal{S}}X_{t,s}\cdot \sum\limits_{p\in\mathcal{P}}V_p(t)\cdot A_p(s)\]
\[\equiv \argmin_X \sum\limits_{t\in\mathcal{T}}\sum\limits_{s\in\mathcal{S}}X_{t,s}\cdot \Big(\abs{\mathcal{P}}-\sum\limits_{p\in\mathcal{P}}V_p(t)\cdot A_p(s)\Big)\]
\[\equiv \argmin_X \sum\limits_{t\in\mathcal{T}}\sum\limits_{s\in\mathcal{S}}X_{t,s}\cdot c_{t,s}\]
}
As $\abs{\mathcal{T}}\leq \abs{\mathcal{S}}$, we can introduce $\big(\abs{\mathcal{S}}-\abs{\mathcal{T}}\big)$ dummy talks $\Big(\mathcal{T}_d=\big\{t'_1,\cdots,t'_{\abs{\mathcal{S}}-\abs{\mathcal{T}}}\big\}\Big)$ with costs: $c_{t,s}=\abs{\mathcal{P}}$, $\forall t,s\in \mathcal{T}_d,\mathcal{S}$.
As scheduling the dummy talks has a constant cost attached, we can now rewrite the above transformed problem as below.
{\small\[\argmin_X \sum\limits_{t\in\mathcal{T}\cup\mathcal{T}_d}\sum\limits_{s\in\mathcal{S}}X_{t,s}\cdot c_{t,s}\]}
This is a minimum cost bipartite matching problem which can be solved in polynomial time using the Hungarian algorithm \cite{kuhn1955hungarian}.
~\\{\bf Proofsketch of Theorem-\ref{theorem:participant_fairness_np_hard}:}
The problem is clearly in NP.
Given a schedule $\Gamma$, participant unfairness can be calculated and verified in $\mathcal{O}\big(m.(n\log n+ l\log l)\big)$ or $\mathcal{O}\big(ml\log l\big)$ time (as $n\leq l$).
To prove NP-hardness, we reduce the {\bf number partitioning problem} (a well-known NP-complete problem \cite{garey1979computers}) to our participant fairness problem. 
An arbitrary instance of number partition problem has a multiset $\mathcal{G}$ of integers, and the task is to decide whether $\mathcal{G}$ can be partitioned into two disjoint subsets $\mathcal{G}_1$ and $\mathcal{G}_2$ such that the sum of numbers in $\mathcal{G}_1$ equals the sum of numbers in $\mathcal{G}_2$.
We provide a polynomial time reduction to a participant fairness instance.
Let $n=\abs{\mathcal{G}}$, $\mathcal{G}=\{g_1,g_2,\cdots,g_n\}$, and $sum(\mathcal{G})=\sum_{i=1}^{n}g_i$.
Now let the set of participants $\mathcal{P}=\{p_1,p_2\}$, and the set of talks $\mathcal{T}=\{t_1,\cdots,t_n\}$.
We set the interest scores of participants as $V_{p_1}(t_i)=V_{p_2}(t_i)=\frac{g_i}{sum(\mathcal{G})}$ $\forall 1\leq i\leq n$.
Let the set of slots be $\mathcal{S}=\{s_1,\cdots,s_n,s_{n+1},\cdots,s_{2n}\}$, and the availability scores of the participants be like: $A_{p_1}(s_i)$ is $1$ for $1\leq i\leq n$ and $0$ for $n+1\leq i\leq 2n$; $A_{p_2}(s_i)$ is $0$ for $1\leq i\leq n$ and $1$ for $n+1\leq i\leq 2n$.
Intuitively, scheduling a talk in the first half of the slots will bring no gain to $p_2$ while scheduling in second half of the slots will bring no gain to $p_1$.
In this setting, it can be easily found that $ICG(p_1)=ICG(p_2)=1$.
Given a polynomial time solution to participant fairness problem (definition-\ref{def:participant_fairness_decision}), we can set $\epsilon=0$, and check if such a conference schedule exists.
Essentially, scheduling talks with $\epsilon=0$ is identical to allocating the set of $n$ talks into (i) half of the slots when $p_1$ is available, and (ii) the other half when $p_2$ is available while ensuring the cumulative gains of $p_1$ and $p_2$ are same (as $ICG$s are already same).
The answer on the existence/non-existence of such a schedule will also answer the existence/non-existence of number partitions with equal sums in polynomial time (as it is a polynomial-time reduction).
Unless P=NP, such a polynomial time solution for the participant fairness problem does not exist.
Thus, the problem is NP-complete.
%
%
~\\{\bf Proof of Claim-\ref{claim:1}:}
	%
	We disprove the negation of the claim using a counter example as given in tab-\ref{tab:toy_example_1}.
	Both participants have interest score of $1$ for the talk.
	Now looking at the availability scores, while $p_1$ and $p_2$ have full ease of availability in $s_1$ and $s_3$ respectively, they both can make themselves available $s_2$ with $0.49$ probability.
	If we consider a efficiency objective (participation maximization) here, we would end up scheduling the talk in either in $s_1$ or in $s_3$;
	if $\Gamma(t)=s_1$ or $\Gamma(t)=s_3$, then $TEP(\Gamma)=1$;
	if $\Gamma(t)=s_2$, then $TEP(\Gamma)=0.98$ which is less.
	However, maximizing participation will either end up with [$NCG(p_1)=1,NCG(p_2)=0$] if $\Gamma(t)=s_1$ or [$NCG(p_1)=0,NCG(p_2)=1$] if $\Gamma(t)=s_3$;
	As both of these results from efficiency optimization provide disparate satisfaction to the participants, they both are unfair.
	On the other hand, if we schedule the talk in $s_2$ ($\Gamma(t)=s_2$), then it becomes fair to both the participants as they will get similar satisfaction [$NCG(p_1)=0.49,NCG(p_2)=0.49$]---they both get a chance to make themselves available in $s_2$ to attend the talk.
	Even though scheduling the talk in $s_2$, ensures participant fairness, it has come at a loss in efficiency; i.e., $TEP(\Gamma)$ reduced from $1$ to $0.98$.
~\\{\bf Proof of Claim-\ref{claim:2}:}
%
We disprove the negation of the claim using a counter example as given in tab-\ref{tab:toy_example_2}.
Now, to maximize efficiency, we can just match talks in decreasing order of overall interest scores to slots in decreasing order of availability scores; i.e., $\Gamma^\text{EM}(t_1)=s_1$ and $\Gamma^\text{EM}(t_2)=s_3$, which will yield $TEP(\Gamma^\text{EM})=1.4$.
The speaker satisfactions for the talks with this schedule will be:
$NEC(t_1|\Gamma^\text{EM})=1$ (as $EC(t_1|\Gamma^\text{EM})=1$ and $IEC(t_1)=1$), and $NEC(t_2|\Gamma^\text{EM})=0.8$ (as $EC(t_1|\Gamma^\text{EM})=0.4$ and $IEC(t_1)=0.5$).
Such disparity in speaker satisfactions can be attributed to speaker unfairness.
In order to reduce speaker-side disparity, we can use a different schedule:
$\Gamma(t_1)=s_3$ and $\Gamma(t_2)=s_2$;
this yields speaker satisfactions $NEC(t_1|\Gamma)=0.8$ and $NEC(t_2|\Gamma)=0.75$ (as $EC(t_1|\Gamma)=0.8$ and $EC(t_2|\Gamma)=0.375$).
This above schedule has in fact the lowest possible disparity in speaker satisfactions, i.e., the highest possible speaker fairness.
Even though this schedule $\{(t_1,s_3),(t_2,s_2)\}$ is fairer to the speakers than the earlier $\{(t_1,s_1),(t_2,s_3)\}$, the gain in speaker fairness has come at a loss in efficiency; $TEP(\Gamma)$ reduced from $1.4$ to $1.175$.
%
~\\{\bf Proof of Claim-\ref{claim:3}:}
%
We disprove the negation of the claim using a counter example as given in tab-\ref{tab:toy_example_3}.
In this example, the schedule $\Gamma=\{(t_1,s_2),(t_2,s_3)\}$ achieves speaker fairness---$NEC(t_1|\Gamma)=NEC(t_2|\Gamma)=0.5$ (as $EC(t_1|\Gamma)=1$, $EC(t_2|\Gamma)=0.7$, while $IEC(t_1)=2$, $IEC(t_2)=1.4$).
However, $\Gamma$ is unfair for the participants---$NCG(p_1|\Gamma)=\frac{1}{1.7}<\frac{0.7}{1.7}=NCG(p_2|\Gamma)$ (as $CG(p_1|\Gamma)=1$, $CG(p_2|\Gamma)=0.7$, while $ICG(p_1)=ICG(p_2)=1.7$).
On the other hand, schedule $\Gamma'=\{(t_1,s_1),(t_2,s_4)\}$ is fair for the participants---$NCG(p_1|\Gamma')=NCG(p_2|\Gamma')=\frac{1.14}{1.7}$, while being unfair for the speakers as $NEC(t_1|\Gamma')=1>0.2=NEC(t_2|\Gamma')$.
%
~\\{\bf Proof of Lemma-\ref{lemma:V_or_A_same}:}
There are three cases where we need to prove that IAM maximizes efficiency;
{\bf (a)} if all participants have identical ease of availability over all available slots $A_p(s)=A(s)$, $\forall s\in \mathcal{S},p\in \mathcal{P}$ (this case is similar to physical conference settings where all participants gather at the same place, thus, have identical ease of availability);
{\bf (b)} if all participants have identical interests over all talks $V_p(t)=V(t)$, $\forall t\in \mathcal{T},p\in \mathcal{P}$;
{\bf (c)} if both (a) and (b) are true.
We, first, reduce the EM objectives in the following cases, and observe that they take a particular form where IAM gives solution.\\
{\bf Case-(a):}
Given $A_p(s)=A(s)$, $\forall s\in \mathcal{S},p\in \mathcal{P}$.
From eq-\ref{eq:TEP}:
\[\text{EM}\equiv \argmax_\Gamma\sum_{p\in\mathcal{P}}\sum_{t\in\mathcal{T}}V_p(t)\times A_p(\Gamma(t))\]
\[\equiv \argmax_\Gamma\sum_{p\in\mathcal{P}}\sum_{t\in\mathcal{T}}V_p(t)\times A(\Gamma(t)) \equiv \argmax_\Gamma\sum_{t\in\mathcal{T}}\sum_{p\in\mathcal{P}}V_p(t)\times A(\Gamma(t))\]
\[\equiv \argmax_\Gamma\sum_{t\in\mathcal{T}} A(\Gamma(t))\Big[\sum_{p\in\mathcal{P}}V_p(t)\Big] \equiv \argmax_\Gamma\sum_{t\in\mathcal{T}} A(\Gamma(t))\times \mathcal{V}(t)\]
(where overall interest score $=\mathcal{V}(t)=\sum_{p\in\mathcal{P}}V_p(t)$)\\
{\bf Case-(b):}
Given $V_p(t)=V(t)$, $\forall t\in \mathcal{T},p\in \mathcal{P}$.
From eq-\ref{eq:TEP}:
\[\text{EM}\equiv \argmax_\Gamma\sum_{p\in\mathcal{P}}\sum_{t\in\mathcal{T}}V_p(t)\times A_p(\Gamma(t))\]
\[\equiv \argmax_\Gamma\sum_{p\in\mathcal{P}}\sum_{t\in\mathcal{T}}V(t)\times A_p(\Gamma(t)) \equiv \argmax_\Gamma\sum_{t\in\mathcal{T}}\sum_{p\in\mathcal{P}}V(t)\times A_p(\Gamma(t))\]
\[\equiv \argmax_\Gamma\sum_{t\in\mathcal{T}} V(t)\Big[\sum_{p\in\mathcal{P}}A_p(\Gamma(t))\Big] \equiv \argmax_\Gamma\sum_{t\in\mathcal{T}} V(t)\times \mathcal{A}(\Gamma(t))\]
(where overall availability score $=\mathcal{A}(s)=\sum_{p\in\mathcal{P}}A_p(s)$)\\
{\bf Case-(c):}
Given $A_p(s)=A(s),V_p(t)=V(t)$, $\forall s\in \mathcal{S},\forall t\in \mathcal{T},p\in \mathcal{P}$.
From eq-\ref{eq:TEP}:
\[\text{EM}\equiv \argmax_\Gamma\sum_{p\in\mathcal{P}}\sum_{t\in\mathcal{T}}V_p(t)\times A_p(\Gamma(t))\]
\[\equiv \argmax_\Gamma\sum_{p\in\mathcal{P}}\sum_{t\in\mathcal{T}}V(t)\times A(\Gamma(t)) \equiv \argmax_\Gamma\sum_{t\in\mathcal{T}}\sum_{p\in\mathcal{P}}V(t)\times A(\Gamma(t))\]
\[\equiv \argmax_\Gamma\sum_{t\in\mathcal{T}}V(t)\times A(\Gamma(t))\]
\noindent In all the three cases above, EM reduces to a form where the terms are independent of individual participant interests and availability while depending only on the overall interest levels ($\mathcal{V}$ in case (a), and participant-independent $V$ in cases (b) \& (c)) and overall availability ($\mathcal{A}$ in case (b), and participant-independent $A$ in cases (a) \& (c)).
Thus, to maximize the reduced objectives in all the cases, top values of overall availability $A$ or $\mathcal{A}$ need to be matched with top values of overall interests $\mathcal{V}$ or $V$ (this also follows from the Rearrangement inequality \cite{hardy1967inequalities})---making it identical to IAM.
{\small
	\begin{algorithm}
		{\raggedright
			{
				{\raggedright{\bf Input :} Set of participants $\mathcal{P}$, set of talks $\mathcal{T}$,   set of slots $\mathcal{S}$ (such that $\abs{\mathcal{T}}\leq\abs{\mathcal{S}}$), the interest scores $V(p,t)$ $\forall p\in \mathcal{P}, t\in \mathcal{T}$, the availability scores $A(p,s)$ $\forall p\in \mathcal{P}, s\in \mathcal{S}$, weight for participant fairness objective $\lambda_1$, and weight for speaker objective $\lambda_2$.}\\
				{\raggedright{\bf Output:} A fair conference schedule $\Gamma$.}
			}
			\caption{{\it RRFS }($\mathcal{P}, \mathcal{T}, \mathcal{S}, V, A, \lambda_1, \lambda_2$)}
			\label{alg:RRFS}
			\begin{algorithmic}[1]
				\STATE $NSch\leftarrow \mathcal{T}$ \COMMENT{Set of all talks which are not yet scheduled}
				\STATE $NAll\leftarrow \mathcal{S}$ \COMMENT{Set of all slots which are not yet allotted}
				\STATE Let $RJO(\mathcal{P}, \mathcal{T}, \mathcal{S}, V, A, \lambda_1, \lambda_2)$ represent the relaxed joint optimization problem (i.e., $X_{t,s}\in \{0,1\}$ in eq-\ref{eq:ip} is replaced with $X_{t,s}\geq 0$).
				
				\WHILE{$NSch\neq \Phi$}
				\STATE{$X\leftarrow RJO(\mathcal{P}, \mathcal{T}, \mathcal{S}, V, A, \lambda_1, \lambda_2)$ \COMMENT{A fractional solution} \label{algstep:RJO}}
				\WHILE{$X$ is not a zero-matrix \COMMENT{Rounding loop}}
				\STATE{$t,s=\argmax_{t',s'\in \mathcal{T},\mathcal{S}}X_{t',s'}$ \COMMENT{Maximum element}}
				\STATE{Set $\Gamma(t)=s$}
				\STATE{$NSch\leftarrow NSch\setminus \{t\}$}
				\STATE{$NAll\leftarrow NAll\setminus \{s\}$}
				\STATE{Set $X_{i,s}=0$ $\forall i\in \mathcal{T}$, and $X_{t,j}=0$ $\forall j\in \mathcal{S}$}
				\ENDWHILE
				\STATE{$\mathcal{T}\leftarrow NSch$, $\mathcal{S}\leftarrow NAll$}
				\ENDWHILE
				
				\STATE Return $\Gamma$.
		\end{algorithmic}}
	\end{algorithm}
}

\if 
{\small
	\begin{algorithm}
		{\raggedright
			{
				{\raggedright{\bf Input :} Set of participants $\mathcal{P}$, set of talks $\mathcal{T}$,   set of slots $\mathcal{S}$ (such that $\abs{\mathcal{T}}\leq\abs{\mathcal{S}}$), the interest scores $V(p,t)$ $\forall p\in \mathcal{P}, t\in \mathcal{T}$, the availability scores $A(p,s)$ $\forall p\in \mathcal{P}, s\in \mathcal{S}$, weight for participant fairness objective $\lambda_1$, and weight for speaker objective $\lambda_2$.}\\
				{\raggedright{\bf Output:} A fair conference schedule $\Gamma$.}
			}
			\caption{{\it RRFS }($\mathcal{P}, \mathcal{T}, \mathcal{S}, V, A, \lambda_1, \lambda_2$)}
			\label{alg:RRFS}
			\begin{algorithmic}[1]
				\State $NSch\leftarrow \mathcal{T}$ \Comment{Set of all talks which are not yet scheduled}
				\State $NAll\leftarrow \mathcal{S}$ \Comment{Set of all slots which are not yet allotted}
				\State Let $RJO(\mathcal{P}, \mathcal{T}, \mathcal{S}, V, A, \lambda_1, \lambda_2)$ represent the relaxed joint optimization problem (i.e., $X_{t,s}\in \{0,1\}$ in eq-\ref{eq:ip} is replaced with $X_{t,s}\geq 0$).
				
				\While {$NSch\neq \Phi$}
				\State $X\leftarrow RJO(\mathcal{P}, \mathcal{T}, \mathcal{S}, V, A, \lambda_1, \lambda_2)$ \Comment{A fractional solution} \label{algstep:RJO}
				\While {$X$ is not a zero-matrix} \Comment{Rounding loop}
				\State $t,s=\argmax_{t',s'\in \mathcal{T},\mathcal{S}}X_{t',s'}$ \Comment{Maximum element}
				\State Set $\Gamma(t)=s$
				\State $NSch\leftarrow NSch\setminus \{t\}$
				\State $NAll\leftarrow NAll\setminus \{s\}$
				\State Set $X_{i,s}=0$ $\forall i\in \mathcal{T}$, and $X_{t,j}=0$ $\forall j\in \mathcal{S}$
				\EndWhile
				\State $\mathcal{T}\leftarrow NSch$, $\mathcal{S}\leftarrow NAll$
				\EndWhile
				
				\State Return $\Gamma$.
		\end{algorithmic}}
	\end{algorithm}
}
\fi 
\begin{table}\footnotesize
	\begin{tabular}{|c|c|c|c|c|c|c|c|c|}
		\hline
		\multirow{2}*{\bf mFairConf} & \multirow{2}*{$NCG_\text{gini}$} & \multirow{2}*{$NCG_\text{mean}$}& \multicolumn{3}{c|}{$NEC_\text{gini}$} & \multicolumn{3}{c|}{$NEC_\text{mean}$}\\\cline{4-9}
		& & & $\mathcal{T}_1$ & $\mathcal{T}_2$ & $\mathcal{T}_3$& $\mathcal{T}_1$ & $\mathcal{T}_2$ & $\mathcal{T}_3$\\\cline{1-9}
		$1$ Full&$0.19$&$0.43$&$0.08$&$0.02$&$0.06$&$0.82$&$0.91$&$0.87$\\\cline{1-9}
		$\mathcal{T}_{123}$ &$0.19$&$0.44$&$0.03$&$0.01$&$0.05$&$0.91$&$0.87$&$0.81$\\\cline{1-9}
		$\mathcal{T}_{1231}$ &$0.17$ &$0.51$ &$0.13$ &$0.01$ &$0.05$ &$1.24$ &$0.87$ &$0.81$\\\cline{1-9}
		$\mathcal{T}_{12312}$ &$0.13$ &$0.56$ &$0.13$&$0.06$ &$0.05$ &$1.24$ &$1.17$ &$0.81$\\\cline{1-9}
		$\mathcal{T}_{123123}$ &$0.10$ &$0.58$ &$0.13$ &$0.06$ &$0.07$ &$1.24$ &$1.17$ &$0.96$\\\cline{1-9}
		$2$ Full & $0.11$& $0.54$& $0.06$& $0.08$& $0.10$&$0.95$ &$1.16$ &$1.18$\\\cline{1-9}
		$\mathcal{T}_{123121}$ &$0.10$ &$0.59$ &$0.13$ &$0.06$ &$0.05$ &$1.37$ &$1.17$ &$0.81$\\\cline{1-9}
	\end{tabular}
	\caption{\small RECSYS priority scheduling results (rounded upto two decimal places). $\mathcal{T}_{1231}$ refers to mFairConf scheduling of top, medium, low, and finally a repetition for top priority talks as discussed in sec-\ref{subsubsec:priority_experiments}. ``$1$ Full" refers to scheduling all talks at once without any grouping. ``$2$ Full" refers to $2$ rounds of full scheduling where the first one is same as ``$1$ Full" and the second one is trying to schedule all talks in the remaining slots as repetitions.}
	\label{tab:priority_scheduling}
	\vspace{-4mm}
\end{table}
~\\{\bf Future Work:}
(i) We limited ourselves in modeling single-track conference scenario;
however, larger conferences often have parallel sessions to accommodate a higher number of talks.
%
The proposed modular framework can be reused with tweaked participant and speaker satisfaction metrics as participants would choose the best talk out of all parallel sessions. The interest score terms (in eqs-\ref{eq:CG},\ref{eq:EC}) will be replaced with maximum over the set of parallel talks in each slot.
(ii) We modeled speaker satisfaction using the expected crowd at her talk, however, the speaker's 
convenience in the assigned time slot could also play a role and can be modeled accordingly;
(iii) Although we considered the participants and speakers to be separate agents in our model, a single agent could be both a participant and a speaker;
thus for such agents with dual roles, the participant satisfaction measure needs to be suitably modified.
(iv) We have limited ourselves to individual fairness for the speakers and participants, without considering their sensitive attributes. However, our joint optimization framework can accommodate additional constraints to include group fairness objectives. 

\end{document}